\documentclass{article}

\usepackage[utf8]{inputenc}
\usepackage{repsize}
\usepackage{latexsym}
\usepackage{defs}
\usepackage{figdefs}
\usepackage{graphicx}
\usepackage{algs}
\usepackage{amsmath,amssymb,amsfonts}
\usepackage{color}

\makeatletter
\date{}
\title{Approximation Algorithms for the Two-Watchman Route in a Simple Polygon}

\author{
Bengt J.~Nilsson\thanks{Malmö University, SE-205~06~~Malmö, Sweden. email:~{\tt bengt.nilsson.TS@mau.se}. ORCID:~0000-0002-1342-8618.
}
\and
Eli Packer\thanks{Yoom, Israel. email:~{\tt eli@yoom.com}.
}
}
\makeatother

\setfigdir{Figures/}
\input{math.tex}

\begin{document}
\maketitle

\begin{abstract}
The {\em two-watchman route problem\/} is that of computing a pair of closed tours in an environment so that the two tours together see the whole environment and some length measure on the two tours is minimized. Two standard measures are: the minmax measure, where we want the tours where the longest of them has smallest length, and the minsum measure, where we want the tours for which the sum of their lengths is the smallest. It is known that computing a minmax two-watchman route is NP-hard for simple rectilinear polygons and thus also for simple polygons. Also, any $c$-approximation algorithm for the minmax two-watchman route is automatically a $2c$-approximation algorithm for the minsum two-watchman route. We exhibit two constant factor approximation algorithms for computing minmax two-watchman routes in simple polygons with approximation factors $5.969$ and $11.939$, having running times $O(n^8)$ and $O(n^4)$ respectively, where $n$ is the number of vertices of the polygon. We also use the same techniques to obtain a $6.922$-approximation for the {\em fixed two-watchman route problem\/} running in $O(n^2)$ time, i.e., when two starting points of the two tours are given as input.
\end{abstract}

\section{Introduction}\label{sec:intro}

Some of the most intriguing problems in computational geometry concern visibility and motion planning in polygonal environments. A classical problem is that of computing a {\em shortest watchman route\/} in an environment, i.e., the shortest closed tour that sees the complete free-space of the environment. Watchman routes can either be {\em fixed}, requiring the tour to pass a given boundary point or {\em floating}, with no requirement to pass any specific point. These problems have been shown NP-hard~\cite{ChiNta:rectwatchman,DumTot:holewatchman} and even $\Omega(\log n)$-inapproximable~\cite{Mit:approxwatchman} for polygons with holes having a total of $n$ segments.

After a sequence of false starts~\cite{ChiNta:watchman,HamNil-SWR:conf,TanHir:watchman,TanHirIna:watchman}, 
Tan~{\em et al}.~\cite{TanHirIna:corrwatchman} prove an $O(n^4)$ time dynamic programming algorithm for computing a shortest fixed watchman route through a given boundary point in a simple polygon. This is later improved to $O(n^3\log n)$ time by Dror~{\em et al}.~\cite{DroEfrLubMit:touringpolygons} and to $O(n^3)$ time by Tan and Jiang~\cite{TanJia:touringpolygons}. Carlsson~{\em et al}.~\cite{CarJonNil:journal} show how to generalize an algorithm that computes a shortest fixed watchman route to compute a shortest floating watchman route in a simple polygon with a quadratic factor overhead. Tan~\cite{Tan:watchman} improves this to a linear factor overhead. Hence, the currently best algorithm for a shortest floating watchman route in a simple polygon uses $O(n^4)$~time.

Given the relatively high polynomial time complexity for computing watchman routes in simple polygons, efficient approximation algorithms are also of interest. Nilsson~\cite{Nil-approxSWR:journal} and Tan~\cite{Tan:approxwatchman} have independently developed linear time approximation algorithms for a shortest floating watchman route in a simple polygon.

The more general problem of computing multiple watchman routes that together see the environment has received much less attention. Mitchell and Wynters~\cite{MitWyn:watchmen} show that already computing the pair of tours that together see a simple rectilinear polygon is NP-hard, if we want to minimize the length of the longest of the two tours, the {\em minmax\/} measure. It is still an open problem whether it is possible to compute a pair of tours for which the sum of the lengths of the two tours is minimal, the {\em minsum\/} measure, in polynomial time. 
Packer~\cite{Pac:multiplewatchmen} gives some experimental results for multiple watchman routes in simple polygons. 
For point sized watchmen, so-called {\em static guards}, Belleville~\cite{Bel:master,Bel:cover} shows an efficiently computable characterization of all simple polygons that are two-guardable with point guards. 

\paragraph*{Our Results.}
We present a polynomial time constant factor approximation algorithm to compute a minmax or minsum pair of tours that together see a simple polygon. We first consider the floating version of the problem and obtain a $5.969$-approximation algorithm for the minmax pair of tours and $11.939$-approximation for the minsum pair of tours that runs in $O(n^8)$ time, where $n$ is the number of vertices of the polygon.  

In the next three sections, we provide some preliminary results and prove some crucial properties that we use continuously in the sequel. In Section~\ref{sec:algorithm}, we give the algorithm for the minmax two-watchman route prove its correctness and analyze its running time. In Section~\ref{sec:tradeoff}, we show how to modify the previous algorithm to run in $O(n^4)$ time while maintaining constant approximation factor, albeit only guaranteeing a factor twice as large as the previous algorithm.
In Section~\ref{sec:fixed} we modify the algorithm to handle the fixed two-watchman route, the case when we have fixed starting points for the tours that they have to pass through, arriving at an $O(n^2)$ time algorithm with approximation factor~$6.922$. We conclude the presentation in Section~\ref{sec:conc}.

\section{Preliminaries}\label{sec:prelim}
Let \P\ be a simple polygon having $n$ vertices and let \dP\ denote the boundary of \P. We say that two points in \P\ {\em see\/} each other, if the line segment connecting the points does not intersect the exterior of \P. For any arbitrary connected object \X\ inside \P, we denote by \VP{\X} the {\em weak visibility polygon\/} of \X\ in \P, i.e., the set of points in \P\ that see some point of~\X.
The boundary of a visibility polygon \VP{\X} consists of edges that are either (sub)edges of \P\ or edges that have their endpoints on \dP\ but their interior points in the interior of \P. These latter edges are denoted the {\em windows\/} of \VP{\X} and they have at least one endpoint on a reflex vertex of \P.
We henceforth assume the existence of linear time algorithms to compute \VP{\X} when \X\ is a point or a segment inside~\P. Such algorithms have previously been presented in the literature~\cite{GooOro:handbook,LeePre:shortpath,SacUrr:handbook}.

A {\em cut\/} is a directed line segment in \P\ with both end points on \dP\ and each interior point is an interior point of \P. Hence, a directed segment incident to a polygon edge or a directed segment intersecting more than two vertices is not a cut. A cut always separates \P\ into exactly two sub-polygons of nonzero area.
If a cut is represented by the segment $[\p,\q]$ we say that the cut is directed from \p\ to \q\ and we call \p\ the {\em start point\/} of the cut and \q\ the {\em end point\/} of the cut. For a cut $c$ in \P, we define the {\em left polygon}, \lhp{c}, to be the set of points in \P\ locally to the left of $c$ according to $c$'s direction.

Assume a counterclockwise walk of \dP. Such a walk imposes a direction on each of the edges of \P\ in the direction of the walk. Consider a reflex vertex of \P. The two edges incident to the vertex can each be extended inside \P\ until the extensions reach a boundary point. These extended segments form cuts given the same direction as the edge they are collinear to. We call these cuts {\em extensions}; see Figure~\ref{defsfigcex}(a) and denote the set of  extensions in \P\ by~\EE.
\begin{figure*} 
			    \begin{center} \small
			    \input{\figdirdefsfigcex.pdf_t}

			    \parcaption{\capw}{\label{defsfigcex}Illustrating the definitions and a counterexample.}
			    \end{center}
			    \end{figure*}

We define a {\em guard set\/} to be any set of points \GG\ that together see all of \P, i.e., $\bigcup_{g\in\GG}\VP{g}=\P$. It is clear that any guard set must have points intersecting \lhp{\e} for every extension \e\ of~\P, since otherwise the edge collinear to \e\ will not be seen by the guard set; see Figure~\ref{defsfigcex}(a). 
Chin and Ntafos~\cite{ChiNta:watchman} prove that this is indeed also a sufficient requirement when the guard set is connected, as it is for a shortest watchman route. For disconnected guard sets, it is easy to construct examples where this requirement is not sufficient; see Figure~\ref{defsfigcex}(b) where the guard set consisting of the three marked convex vertices has points to the left of each essential extension but it does not see the complete polygon. 

Let $c$ be a cut. If a guard set \GG\ intersects \lhp{c}, we say that $c$ is {\em covered\/} by \GG. Furthermore, if \GG\ intersects the interior of \lhp{c}, then \GG\ {\em properly covers\/} $c$. If \GG\ properly covers $c$ and intersects $c$, we say that \GG\ {\em crosses\/} $c$. Finally, if \GG\ covers $c$, but does not properly cover $c$, then \GG\ {\em reflects\/} on~$c$. 


We also make use of the fact that shortest paths in \P\ between combinations of segments and points can be computed efficiently~\cite{GooOro:handbook,GuiHerLevShaTar:visandshortpath,LeePre:shortpath}. We denote the shortest path between two objects $X$ and $Y$ in \P\ by~\SP(X,Y).

Let \Xa\ and \Xb\ be two closed polygonal cycles contained in a simple polygon \P, such that each point in \P\ sees some point on \Xa\ or \Xb. We call the pair $\XX=(\Xa,\Xb)$, a {\em two-watchman route}. The length of a cycle \X\ in \P\ is denoted \len{\X} and we let $\sumlen{\XX}\defeq\len{\Xa}+\len{\Xb}$ be the {\em sum length\/} of \XX\ and $\maxlen{\XX}\defeq\max\big\{\len{\Xa},\len{\Xb}\big\}$ be the {\em max length\/} of~\XX.

Let $\SS=(\Sa,\Sb)$ and $\TT=(\Ta,\Tb)$ be two two-watchman routes such that $\sumlen{\SS}\leq\sumlen{\XX}$ and $\maxlen{\TT}\leq\maxlen{\XX}$ for any two-watchman route \XX\ in~\P.
We say that \SS\ is a {\em minsum\/} two-watchman route and \TT\ is a {\em minmax\/} two-watchman route.
The following inequalities are immediate from the definitions, for any two-watchman route \XX,
\begin{align}
\maxlen{\XX}=&\max\big\{\len{\Xa},\len{\Xb}\big\}\leq\len{\Xa}+\len{\Xb}=\sumlen{\XX}, \label{math:maxleqsum} \\
\sumlen{\XX}=&\len{\Xa}+\len{\Xb}\leq2\max\big\{\len{\Xa},\len{\Xb}\big\}=2\maxlen{\XX}, \label{math:sumleq2max}
\end{align}
and therefore
\begin{equation}\label{math:sleq2tleq2s}
\maxlen{\TT} \leq \sumlen{\SS} \leq 2\maxlen{\TT}.
\end{equation}
Hence, computing a $c$-approximation for one measure also gives at most a $2c$-approx\-ima\-tion for the other measure.

One could imagine that there always is a cut in the polygon such that the two tours that are shortest watchman tours for the two sub-polygons formed are almost as short as the optimum two-watchman tour. Trying out all possible cuts, and solving the two single watchman tour subproblems in each case would give a simple and efficient algorithm with good approximation ratio. One would then imagine wrong, as the counterexample in Figure~\ref{defsfigcex}(c) shows. The red and blue point sized tours, each see the grey region and the regions of their own color. Here, any partition of the polygon into two pieces by a cut will make the watchman tour solution of each piece infinitely longer than the optimal solution. The example can easily be modified for non-point sized tours.


\section{General Properties of Two-Watchman Routes}\label{sec:properties}

We make the following assumption about the polygons considered in this section.
\begin{assumption}\label{ass:notguardable}
The polygon \P\ considered in this section is simple and not guardable by one or two point guards. Hence, \P\ is not convex or starshaped.
\end{assumption}

The following lemma is at the heart of our construction of an algorithm to compute a two-watchman route.
\begin{lemma}\label{lem:boundary}
If two tours in \P\ see all of\/ \dP, then they see all of\/~\P.
\end{lemma}
\begin{proof}
We do a proof by contradiction. Let \Xa\ and \Xb\ be two tours in \P\ and assume that \p\ is an interior point of \P\ not seen by any of them. We show that there must be some boundary point that is also not seen by the tours contradicting that \Xa\ and \Xb\ together see \dP. Let \VP{\Xa} and \VP{\Xb} be the two visibility polygons of the tours. Each of them must have boundary segments \wa\ and \wb\ separating the subpolygon containing \p\ from the subpolygon containing the tour. The segment \wa\ cuts \P\ into two subpolygons. Let \Pxa\ be the subpolygon containing \Xa\ and let \Pp\ be the subpolygon containing \p. If \Xb\ intersects \wa, we interchange the roles of \Xa\ and \Xb, thus we can assume that \Xb\ is either completely in \Pxa\ or completely in \Pp. We have three cases:
\begin{description}
\item[\Xb\ lies in \Pxa.]
Let \SP(\p,\Xb) be the shortest path from \p\ to \Xb\ and let \s\ be the first segment of this path. It connects \p\ with a reflex vertex \v\ on the boundary of \P. Extend \s\ away from \v\ until it hits the boundary at \p'; see Figure~\ref{boundaryproof}(a). The point \p'\ is clearly not seen by \Xb\ since if it were, then \p\ would also bee seen by \Xb. The point \p'\ is not seen by \Xa\ either since \SP(\p,\Xb) crosses \wa\ at some point \q\ with \SP(\p,\q) in \Pp\ and since \SP(\p,\q) is a shortest path, the extension from \p\ to \p'\ cannot cross~\wa; see Figure~\ref{boundaryproof}(b).
\item[\Xb\ lies in \Pp\ and \wa\ and \wb\ intersect.]
Let \q\ be the intersection point between \wa\ and \wb\ and let \s\ be the first segment of \SP(\p,\q). It connects \p\ either with a reflex vertex \v\ on the boundary of \P\ or with \q. Extend \s\ away from \v\ (or \q) until it hits the boundary at \p'; see Figure~\ref{boundaryproof}(c). The point \p'\ cannot be seen by \Xa\ or \Xb\ since \SP(\p,\q) lies outside both \VP{\Xa} and \VP{\Xb} (except for the point \q), and hence, since \SP(\p,\q) is a shortest path, the extension from \p\ to \p'\ cannot cross any of \wa\ or~\wb.
\item[\Xb\ lies in \Pp\ and \wa\ and \wb\ do not intersect.]
It is clear that there exist points on the boundary not seen by the two tours since the boundary points close to the end points of \wa\ not in \Pxa\ and the boundary points close to the end points of \wb\ not in \Pxb\ are not seen by any of the tours; see Figure~\ref{boundaryproof}(d).
\end{description}
This completes the proof.
\begin{figure*} 
			    \begin{center} \small
			    \input{\figdirboundaryproof.pdf_t}

			    \parcaption{\capw}{\label{boundaryproof}Illustrating the proof of Lemma~\protect\ref{lem:boundary}.}
			    \end{center}
			    \end{figure*}
\end{proof}
The lemma implies that it is sufficient that our algorithm constructs two tours that together see the whole boundary of \P\ to guarantee that all of \P\ is guarded.

Consider two tours \Xa\ and \Xb\ and a polygon boundary edge~\be.
\begin{lemma}\label{lem:edgeconnected}
For any two tours \Xa\ and \Xb\ and a polygon boundary edge \be, the sets $\VP{\Xa}\cap\be$, $\VP{\Xb}\cap\be$ and $\VP{\Xa}\cap\VP{\Xb}\cap\be$ are each connected.
\end{lemma}
\begin{proof}
If the set under consideration is empty, the lemma follows by definition so assume in each case that it is nonempty. If any of $\VP{\Xi}\cap\be$, for $i\in\{1,2\}$, is connected, it is a subsegment of \be\ and therefore convex. Since the intersection of two connected and convex sets is also a connected and convex set, it follows that the connectedness of $\VP{\Xa}\cap\VP{\Xb}\cap\be=\bigcap_{i\in\{1,2\}} \VP{\Xi}\cap\be$ is an immediate consequence, if we can show that the two sets $\VP{\Xi}\cap\be$, for $i\in\{1,2\}$, are connected.

We make a proof by contradiction and assume that $\VP{\Xi}\cap\be$ is disconnected, for each $i\in\{1,2\}$. By following \be\ from one end point to the other we pass each of the connected components giving us an ordering of them. Let \ra\ be a point in the first component and let \rb\ be a point in the last component. Each point \ra\ and \rb\ is seen by \Xi. Hence, for \Xi\ there is a point \rr'\ on \be\ between \ra\ and \rb\ that is not seen by \Xi. Let \pj, for $j\in\{1,2\}$, be two points on \Xi\ that see \rj\ respectively. Without loss of generality, we can assume that $[\pj,\rj]$ does not intersect \Xi\ except at \pj, for $i\in\{1,2\}$. If the segments $[\pa,\ra]$ and $[\pb,\rb]$ do not intersect, we construct a closed simple polygon \Ri\ in \P\ as follows: follow \be\ from \ra\ to \rb, from \rb\ to \pb, from \pb\ to \pa\ along a simple path in \Xi, and finally from \pa\ to \ra. The only edges of the constructed polygon \Ri\ that are not part of \Xi\ are the three edges adjacent to \ra\ and \rb\ so \Ri\ is completely seen by the part of the tour \Xi\ connecting \pb\ to \pa. Hence, the point \rr'\ is also seen since it lies on \be\ between \ra\ and \rb, giving us a contradiction.

On the other hand if the segments $[\pa,\ra]$ and $[\pb,\rb]$ intersect, let \q\ be an intersection point between $[\pa,\ra]$ and $[\pb,\rb]$. The three points \pa, \pb, and \q\ form a triangle interior to \P\ and we construct a closed simple polygon \Rip\ in \P\ having \q\ as a vertex as follows: follow a simple path of \Xi\ from \pa\ to \pb, a straight edge from \pb\ to \q, and a straight edge from \q\ to \pa. If we extend the segment $[\p',\q]$ inside \Rip, it must intersect \Xi, since \pa\ connects to \pb\ in \Rip\ using a simple path of \Xi. Hence, the point \p'\ is also seen from \Xi, again giving us a contradiction.
\end{proof}

From Lemma~\ref{lem:edgeconnected} we have that \Ta\ and \Tb\ both see connected components of any boundary edge of \P. Hence, a boundary edge \be\ can be partitioned into at most three subsegments, at most one of which is seen by both \Ta\ and \Tb\ and the remaining at most two are seen by one of \Ta\ and~\Tb.

Our next lemma shows that $\maxlen{(\Ta,\Tb)}<\len{\Ws}$ and this strict inequality is used in the proof of Lemma~\ref{lem:crossing}.
\begin{lemma}\label{lem:smaller}
Let $(\Ta,\Tb)$ be the shortest minmax two-watchman route and \Ws\ the shortest watchman route. If $\len{\Ws}>0$,
$$\maxlen{(\Ta,\Tb)}<\len{\Ws}.$$
\end{lemma}
\begin{proof}
Consider a shortest watchman route \Ws\ and let \p\ and \q\ be two points on \Ws\ such that if you follow the tour \Ws\ from \p\ a distance of $\len{\Ws}/2$ in counterclockwise order, you reach the point \q. It is clear that following the tour from \p\ to \q\ in clockwise order also gives a path of length $\len{\Ws}/2$.

Let \SP(\p,\q) be the shortest path from \p\ to \q\ in \P. If $\len{\SP(\p,\q)}<\len{\Ws}/2$, we construct a two-watchman route $(\Xa,\Xb)$ such that $\maxlen{(\Xa,\Xb)}<\len{\Ws}$ as follows. Let \Xa\ be the tour obtained by following \SP(\p,\q) from \p\ to \q\ and \Ws\ from \q\ to \p\ in counterclockwise order. Since $\len{\SP(\p,\q)}<\len{\Ws}/2$, the length of \Xa\ is strictly smaller than \len{\Ws}. The tour \Xb\ is obtained by following \Ws\ from \p\ to \q\ in counterclockwise order and then \SP(\p,\q) from \q\ to \p. Again, the length of \Xb\ is strictly smaller than \len{\Ws}. 

If $\len{\SP(\p,\q)}=\len{\Ws}/2$, on the other hand, then both the clockwise and counterclockwise paths from \p\ to \q\ along \Ws\ follow \SP(\p,\q), since shortest paths are unique. We construct a two-watchman route $(\Xa,\Xb)$ as follows. Let \rr\ be the midpoint on \SP(\p,\q) and let \Xa\ be the tour obtained by following \SP(\p,\q) from \p\ to \rr\ and back to \p. Similarly, let \Xb\ be the tour obtained by following \SP(\p,\q) from \q\ to \rr\ and back to \q. We have that $\maxlen{(\Xa,\Xb)}=\len{\Ws}/2<\len{\Ws}$ also in this case.

Since $(\Ta,\Tb)$ is the shortest minmax two-watchman route, it has max-length bounded by that of $(\Xa,\Xb)$ and is hence also strictly smaller than~\len{\Ws}, concluding the proof.
\end{proof}%

In the previous section, we defined \EE\ to be the set of extensions of the edges in \P. There is a subdivision of \EE\ into nonempty subsets \Ea\ and \Eb, such that each tour \Ta\ and \Tb\ of a minmax two-watchman route covers the extensions in \Ea\ and \Eb\ respectively. It is clear that neither \Ea\ nor \Eb\ can be empty since if one of them is empty, the other contains all the extensions of \EE. This means that one of \Ta\ or \Tb\ covers all the extensions in \EE, but the shortest tour that covers all extensions in \EE\ is \Ws, the shortest watchman route, contradicting Lemma~\ref{lem:smaller}. Note also that the subdivision is not necessarily a partition since an extension in \EE\ can be covered by both \Ta\ and \Tb. The following stronger claim also holds.

\begin{lemma}\label{lem:crossing}
There exists a minmax two-watchman route $\TT=(\Ta,\Tb)$, such that each tour \Ti\ {\em intersects\/} some extension in~\Ei,~$i\in\{1,2\}$.
\end{lemma}
\begin{proof}
Assume that both tours \Ta\ and \Tb\ are as short as possible. At least one of them will have the length \maxlen{(\Ta,\Tb)} and the other the shortest possible length given that the first tour achieves the length~\maxlen{(\Ta,\Tb)}. We will show the result only for tour \Ta, but the same argument holds also for~\Tb.

Assume first that $\len{\Ta}=0$, i.e., \Ta\ is a point sized tour. This implies that $\P\setminus\!\VP{\Tb}$ is starshaped even though it may be a disconnected set. The set of points in \P\ that see all of $\P\!\setminus\!\VP{\Tb}$, also known as the {\em kernel}, is the intersection of the left halfplanes collinear to the boundary edges from $\dP\cap(\P\!\setminus\!\VP{\Tb})$~\cite{LeePre:kernel}. Left means here locally to the left of the associated edge in the order of the counterclockwise traversal. The kernel boundary cannot be made up only of boundary edges of \P, since that would make \P\ a convex polygon and thus guardable with one point guard contradicting Assumption~\ref{ass:notguardable}. Hence, there is a kernel boundary edge collinear with some extension of \P\ and we can let \Ta\ be a point on this kernel boundary edge.

Now, assume that $\len{\Ta}>0$. We make a proof by contradiction and assume further that \Ta\ does not intersect any extension in \EE. If this is the case, \Ta\ cannot have any reflex vertices, since if \Ta\ does, then any such vertex coincides with a reflex vertex of \P\ and thus \Ta\ intersects the two extensions adjacent to this vertex. Hence, \Ta\ has only convex vertices. If \Ta\ is a line segment, we consider the two endpoints of the segment to be the convex vertices of the tour that goes back and forth between them. 

Consider the points of the boundary \dP\ of \P\ that are seen by \Ta. The visibility from \Ta\ subdivides \dP\ into disjoint (maximal) subpaths and we color the interior points of each subpath {\em white\/} if \Ta\ sees these points and {\em black\/} if \Ta\ does not see the points. The endpoints of each (maximal) subpath is colored {\em grey}. For a color $c\in\{\mbox{\it black}, \mbox{\it grey}, \mbox{\it white}\}$, we say that a point on \dP\ has color $c$ {\em for \Ta}. We can similarly color the boundary for \Tb\ in which case we say that a point has color $c'$ for~\Tb, $c'\in\{\mbox{\it black}, \mbox{\it grey}, \mbox{\it white}\}$.

We refine the coloring of the boundary somewhat by considering the convex vertices of \Ta. Let \u\ be a convex vertex of \Ta. There exists at least one grey boundary point $\p(\u)$ that is seen from \u\ but not from any other point of \Ta. The point $\p(\u)$ must exist, otherwise \Ta\ can be made shorter, contradicting that both tours \Ta\ and \Tb\ are as short as possible. In fact, \u\ can have many such points. We therefore consider a convex vertex \u\ of \Ta\ to have {\em multiplicity} $k$, if there are $k$ different points $\p(\u)$ associated to \u. We differentiate between $\p(\u)$ and $\p(\u')$ even though $\u=\u'$ and has multiplicity at least two. We let the color of each point $\p(\u)$ be {\em dark grey}. The remaining grey points are considered to be {\em light grey}. 

Let $\lx(\u)$ be the maximal line segment in \P\ passing through the points \u\ and $\p(\u)$. The segment $\lx(\u)$ must intersect at least one reflex vertex of \dP\ between \u\ and $\p(\u)$ that hides $\p(\u)$ from $\Ta\!\setminus\!\u$. Let $\v(\u)$ be the reflex vertex of \dP\ on $\lx(\u)$ closest to $\p(\u)$. The line segment $[\p(\u),\v(\u)]$ partitions \P\ into two subpolygons \Pa\ containing \Ta\ and \Pz. Let $\be(\u)$ be the polygon boundary edge adjacent to the reflex vertex $\v(\u)$ in \Pz\ and let $\ex(\u)$ be the extension collinear to $\be(\u)$ inside \Pa. We associate a direction to $\lx(\u)$ so that $\ex(\u)$ lies locally to the left of~$\lx(\u)$. 
Henceforth, we denote \Pz\ by \Peu\ and \Pa\ by \Plu\ to let them depend on the vertex \u\ of~\Ta.
We refer to Figure~\ref{crossingfig}(a) for an illustration of the given definitions.%
\begin{figure*} 
			    \begin{center} \small
			    \input{\figdircrossingfig.pdf_t}

			    \parcaption{\capw}{\label{crossingfig}Illustrating the proof of Lemma~\protect\ref{lem:crossing}.}
			    \end{center}
			    \end{figure*}

The argument for a contradiction relies on the following claim: ``any boundary point that is dark grey for \Ta\ must also be grey for \Tb.'' Assuming boundary point \p\ is dark grey for \Ta, we make the following case analysis.
\begin{description}
\item 
If \p\ is white for \Tb, we have an immediate contradiction since the single convex vertex \u\ of \Ta\ that sees \p\ can be cut away from \Ta\ arbitrarily close to \u, thus shortening the length of \Ta; see Figure~\ref{crossingfig}(b).
\item
If \p\ is black for \Tb, this means that there is open interval on the boundary \dP\ centered at \p\ that is not seen by any point of \Tb. Since \p\ is one endpoint of a maximal subpath that is seen by \Ta, there exist boundary points not seen by either \Ta\ or \Tb, a contradiction.
\end{description}
The rest of the proof now shows that, if \Ta\ does not intersect some extension, then there exists a dark grey boundary point for \Ta\ that is either black or white for \Tb, or \Ta\ can be shortened, thus establishing a contradiction. We do this by a further case analysis.

\begin{description}
\item 
If \Ta\ has two vertices \u\ and \u'\!, such that $\lhpb{\lx(\u)}\cap\lhpb{\ex(\u')}=\emptyset$, then the intersection $\lhpb{\ex(\u)}\cap\lhpb{\ex(\u')}$ is also empty, since $\Big(\lhpb{\ex(\u)}\!\setminus\!\Peu\Big)\subset\lhpb{\lx(\u)}$. This means that \Ta\ lies properly in the region $\P\setminus\Big(\lhpb{\ex(\u)}\cup\lhpb{\ex(\u')}\Big)$ and \Tb\ must intersect the two disjoint regions $\lhpb{\ex(\u)}$ and $\lhpb{\ex(\u')}$. Since $\lhpb{\lx(\u)}\cap\lhpb{\ex(\u')}=\emptyset$, \Tb\ must cross $\lx(\u)$ and hence $\p(\u)$ is white for \Tb\ giving us a contradiction; see Figures~\ref{crossingfig}(a) and~(b).
\item
If all pairs of vertices \u\ and \u'\ of \Ta\ (with multiplicity) have the property that $\lhpb{\lx(\u)}\cap\lhpb{\ex(\u')}\neq\emptyset$, then since $\Big(\lhpb{\ex(\u')}\setminus\Peup\Big)\subset\lhpb{\lx(\u')}$ it also holds that $\lhpb{\lx(\u)}\cap\lhpb{\lx(\u')}\neq\emptyset$ and we have further subcases.
\begin{description}
\item
If $\bigcap_{\u\in\UU}\lhpb{\lx(\u)}\neq\emptyset$, where \UU\ is the set of vertices of \Ta\ (with multiplicity), then let \q\ be a point in $\bigcap_{\u\in\UU}\lhpb{\lx(\u)}$ and let \ua\ and \ub\ be two distinct vertices of \Ta\ such that \Ta\ is completely contained in the cone $\angle\ua,\q,\ub$. We note that \Ta\ is the shortest tour that visits the regions \lhpb{\lx(\u)}, for $\u\in\UU$. 

Let $\bigtriangleup\ua,\q,\ub$ denote the triangle with corners at \ua, \q, and \ub. The interior region of $\angle\ua,\q,\ub\setminus\!\bigtriangleup\ua,\q,\ub$ cannot contain convex vertices of \Ta, since any such vertex \u'\ would have $\q\not\in\lhpb{\lx(\u')}$ contradicting that $\bigcap_{\u\in\UU}\lhpb{\lx(\u)}\neq\emptyset$; see Figure~\ref{crossingfig}(c). \Ta\ thus connects \ua\ and \ub\ by a line segment.

Vertex \ua\ (with multiplicity) intersects $\bigcap_{\u\in\Uua}\!\!\lhpb{\lx(\u)}$ for some subset \Uua\ of \UU. Similarly, the vertex \ub\ (with multiplicity) intersects $\bigcap_{\u\in\Uub}\!\!\lhpb{\lx(\u)}$ for some subset \Uub\ of \UU, with $\Uua\!\cap\Uub=\emptyset$.
Since the intersection of locally convex regions is also locally convex, any point on the line segment $[\ua,\q]$ intersects $\bigcap_{\u\in\UU_{\ua}}\!\lhp{\lx(\u)}$ and similarly, any point on the line segment $[\ub,\q]$ intersects $\bigcap_{\u\in\UU_{\ub}}\!\!\lhpb{\lx(\u)}$. Thus we can move \ua\ and \ub\ (with multiplicity) along their corresponding line segments towards \q, thus shortening \Ta, giving us a contradiction; see Figure~\ref{crossingfig}(c).

\item 
If $\bigcap_{\u\in\UU}\lhpb{\lx(\u)}=\emptyset$ and for all pairs \u\ and \u', $\lhpb{\lx(\u)}\cap\lhpb{\lx(\u')}\neq\emptyset$, \Ta\ must have at least three vertices. We first show that there exists a subset of three vertices \ua, \ub, and \uc\ such that $\lhpb{\lx(\ua)}\cap\lhpb{\lx(\ub)}\cap\lhpb{\lx(\uc)}=\emptyset$.

Since $\Ta\cap\lhpb{\lx(\u)}=\u$, for every vertex \u\ of \Ta, not only do $\lhpb{\lx(\u)}\cap\lhpb{\lx(\u')}\neq\emptyset$ but the segments $\lx(\u)$ and $\lx(\u')$ intersect in a point, for every vertex pair \u\ and \u'\!.  Pick \ua\ to be any vertex of \Ta\ and assume without loss of generality that $\lx(\ua)$ is horizontal and directed towards the right; see Figure~\ref{crossingfig}(d). Initialize the set \LL\ to be $\LL:=\{\ua\}$ and sort the remaining vertices \u\ on the angle the corresponding cut $\lx(\u)$ makes with $\lx(\ua)$, from 0 to $2\pi$. Now, add vertices \u\ (with multiplicity), one by one, according to the sorted order to \LL, for $\u\in\UU$ until $\bigcap_{\u\in\LL}\lhpb{\lx(\u)}=\emptyset$. Let \ub\ denote the  last vertex added to \LL\ during the process above, and let \uc\ be a vertex in \LL\ such that the intersection $\qa=\lx(\uc)\cap\lx(\ua)$ lies after $\qb=\lx(\uc)\cap\lx(\ub)$ on the  directed cut \lx(\uc). To see that $\lhpb{\lx(\ua)}\cap\lhpb{\lx(\ub)}\cap\lhpb{\lx(\uc)}=\emptyset$, any point in the intersection $\lhpb{\lx(\ua)}\cap\lhpb{\lx(\uc)}$ lies above (or on)  \lx(\ua) and any point in the intersection $\lhpb{\lx(\ub)}\cap\lhpb{\lx(\uc)}$ lies below \lx(\ua), since \qb\ lies below~\lx(\ua); see Figure~\ref{crossingfig}(d). This also shows that \uc\ must exist, since if, for all vertices $\u\in\LL\!\setminus\!\{\ua,\ub\}$, the point $\lx(\u)\cap\lx(\ua)$ lies before $\lx(\u)\cap\lx(\ub)$ on $\lx(\u)$, the first such intersection point on $\lx(\ub)$ along $\lx(\ub)$ lies to the left (or on) each cut $\lx(\u)$, for $\u\in\LL$, contradicting that $\bigcap_{\u\in\LL}\lhpb{\lx(\u)}=\emptyset$; see Figure~\ref{crossingfig}(d).

Since $\lhpb{\lx(\ua)}\cap\lhpb{\lx(\ub)}\cap\lhpb{\lx(\uc)}=\emptyset$, it also follows that $\lhpb{\ex(\ua)}\cap\lhpb{\ex(\ub)}\cap\lhpb{\ex(\uc)}=\emptyset$, as $\Big(\lhpb{\ex(\u)}\!\setminus\!\Peu\Big)\subset\lhpb{\lx(\u)}$, for every vertex \u\ of \Ta.

The tour \Ta\ has no points in any of the regions \lhpb{\ex(\ua)}, \lhpb{\ex(\ub)}, or \lhpb{\ex(\uc)}, so \Tb\ must intersect each of them. Assume that \Tb\ does not intersect $\lx(\ub)$ or $\lx(\uc)$, otherwise at least one of $\p(\ub)$ or $\p(\uc)$ is white for \Tb, giving us a contradiction. However, then \Tb\ must have interior points in $\lhpb{\lx(\ub)}\cap\lhpb{\lx(\uc)}$ since \Tb\ intersects \lhpb{\ex(\ub)} and \lhpb{\ex(\uc)} but does not intersect $\lx(\ub)$ or $\lx(\uc)$. Hence, \Tb\ has points below $\lx(\ua)$ and since \lhpb{\ex(\ua)} lies above (or on) $\lx(\ua)$, \Tb\ must intersect $\lx(\ua)$ and thus $\p(\ua)$ is white for \Tb, again giving us a contradiction; see Figure~\ref{crossingfig}(d).
\end{description}
\end{description}
This concludes the proof.
\end{proof}

\section{Tentacles and Jellyfish}\label{sec:jellyfish}

We will use the following definitions extensively in the sequel.
\begin{definition}
For a point \q\ in \P\ and a point \rr\ on the boundary \dP\!, we call the shortest path from \q\ to some point in \P\ that sees \rr, a {\em tentacle\/} from \q\ to \rr, denoted~\ZT(\q,\rr). We say that \q\ is the {\em head\/} of the tentacle and that a tentacle is {\em attached\/} to its head. The other endpoint of the tentacle is called the~{\em tip}.
\end{definition}

With a tentacle \ZT(\q,\rr) (where \q\ does not see \rr) we also associate a {\em tentacle cut} \cut{\ZT(\q,\rr)}. Consider the maximal line segment $l$ in \P\ passing through the tip \p\ of \ZT(\q,\rr) and the boundary point \rr. If the segment $l$ has endpoints \rr\ and \rr'\ and (possibly) subdivides into connected pieces $l_1,l_2,\ldots$ intersecting the boundary only at the end points and where $l_i$ partitions \P\ into two subpolygons, one containing \q\ and one containing \rr. The cut \cut{\ZT(\q,\rr)} is the segment $l_i$ directed so that $\q\in\P\setminus\lhpb{\cut{\ZT(\q,\rr)}}$; see Figure~\ref{tentacle}(a). The first boundary vertex \u\ encountered as you move from \rr\ along $l$ towards \p\ is the {\em hiding vertex\/} of the tentacle~\ZT(\q,\rr).

We can prove the following technical lemma.
\begin{lemma}\label{lem:tentacle}
If \q\ is a point in \P\ and \sbe\ is a subsegment of a boundary edge \be\ of \dP\ having endpoints \ra\ and \rb, then the two tentacles \ZT(\q,\ra) and \ZT(\q,\rb) together see the whole subsegment \sbe. 
\end{lemma}
\begin{proof}
The proof is a simple modification of the proof of Lemma~\ref{lem:edgeconnected}. 
Let \pa\ and \pb\ be the two tips of \ZT(\q,\ra) and \ZT(\q,\rb), respectively. Let \rr\ be some arbitrary point in \sbe. If the segments $[\pa,\ra]$ and $[\pb,\rb]$ do not intersect, we construct a closed simple polygon \RR\ in \P\ as follows: follow \sbe\ from \ra\ to \rb, from \rb\ to \pb, from \pb\ to \q\ along the tentacle \ZT(\q,\rb), from \q\ to \pa\ along \ZT(\q,\ra), and finally from \pa\ to \ra. The only edges of the constructed polygon \RR\ that are not part of the two tentacles are the three edges adjacent to \ra\ and \rb\ so \RR\ is completely seen by the two tentacles. Hence, the point \rr\ is also seen since it lies on \sbe\ between \ra\ and \rb.

On the other hand if the segments $[\pa,\ra]$ and $[\pb,\rb]$ intersect, let \p\ be an intersection point between $[\pa,\ra]$ and $[\pb,\rb]$. The three points \ra, \rb, and \p\ form a triangle interior to \P\ and we construct a closed simple polygon \Rp\ in \P\ having \p\ as a vertex as follows: follow the tentacles from \ra\ to \rb\ via \q, a straight edge from \rb\ to \p, and a straight edge from \p\ to \ra. If we extend the segment $[\rr,\p]$ inside \Rp, it must intersect a tentacle, since \ra\ connects to \rb\ in \Rp\ using the tentacles. Hence, the point \rr\ is also seen from the tentacles.
\end{proof}
Let $s$ be a line segment in \P\ and let $\be=[\v,\v']$ be some boundary edge of \dP. The next lemma establishes that if we move the head of the tentacle \ZT(\q,\rr), from $\q\in s$, a small distance to the point $\q'\in s$ and the point $\rr\in]\v,\v'[$ to the point $\rr'\in]\v,\v'[$ also a small distance, the difference in length between the two tentacles $\|\ZT(\q,\rr)\|-\|\ZT(\q',\rr')\|$ is a smooth function.%
\begin{figure*} 
			    \begin{center} \small
			    \input{\figdirtentacle.pdf_t}

			    \parcaption{\capw}{\label{tentacle}(a)~Illustrating the tentacle \protect\ZT(\protect\q,\protect\rr)\!, its hiding vertex, and its associated cut. (b)~The difference between tentacle and edge restricted tentacle. Vertex \protect\u\ is the hiding vertex for \protect\ZT(\protect\q,\protect\v) and \protect\v\ is the hiding vertex for~\protect\ZR(\protect\q,\protect\v,\protect\be).}
			    \end{center}
			    \end{figure*}
\begin{lemma}\label{lem:tentaclemotion}
Let \q\ move a distance $\delta$ to \q'\ on a line segment $s$ and let \rr\ move a distance $\epsilon$ to \rr', where both \rr\ and \rr'\ lie in the open interval $]\v,\v'[$ of a boundary edge $\be=[\v,\v']$, in such a way that the first segment of the tentacles from \q\ and \q'\ intersect the same reflex vertex, if the tentacle consists of multiple segments, and \cut{\ZT(\q,\rr)} and \cut{\ZT(\q',\rr')} have the same hiding vertex, then
$$
\len{\ZT(\q',\rr')}  
=
\len{\ZT(\q,\rr)} +
{\cal F}(\delta,\epsilon),
$$
such that 
\begin{align*}
{\cal F}(\delta,\epsilon) & = 
- F_0 + \sqrt{F_0^2 + F_1\delta + F_2\delta^2}
- F_3 + \frac{F_3+F_4\epsilon+F_5\delta+F_6\epsilon\delta+F_7\epsilon^2+F_8\epsilon^2\delta}{\sqrt{1+F_9\epsilon+F_{10}\epsilon^2+F_{11}\epsilon^3+F_{12}\epsilon^4}}
\nonumber\\
& \qquad
- F_{13} + \sqrt{\frac{F_{13}^2+F_{14}\epsilon+F_{15}\delta+F_{16}\epsilon\delta+F_{17}\epsilon^2+F_{18}\delta^2+F_{19}\epsilon^2\delta+F_{20}\epsilon\delta^2+F_{21}\epsilon^2\delta^2}{1+F_{22}\epsilon+F_{23}\epsilon^2}},
\end{align*}
where $F_0,\dots,F_{23}$ are constants.
\end{lemma}
The proof of the lemma is a lengthy case analysis where all arguments are based on similarity and the cosine theorem and is therefore deferred to Appendix~\ref{app:tentaclemotion}.

To alleviate the fact that Lemma~\ref{lem:tentaclemotion} only holds for points \rr\ in the interior of boundary edges we define the {\em edge restricted tentacle\/} \ZR(\q,\rr,\be) to be
\begin{equation}
\ZR(\q,\rr,\be) = \!\!\!\!\lim_{]\v,\v'[\ni\p\rightarrow\rr} \!\ZT(\q,\p)\!, \mbox{\ where $\be=[\v,\v']$.}
\end{equation}
An edge restricted tentacle \ZR(\q,\rr,\be) can differ from \ZT(\q,\rr) only when \rr\ is a vertex of \be. If \v\ is a reflex vertex and \q\ lies to the right of extension \e\ collinear to boundary edge \be, then \ZR(\q,\v,\be) is the shortest path from \q\ that sees \be, not just \v; see Figure~\ref{tentacle}(b). In all other cases, $\ZR(\q,\rr,\be)=\ZT(\q,\rr)$. The proof of Lemma~\ref{lem:tentacle} does not make use of the edge restriction of a tentacle \ZR(\q,\rr,\be) so it still holds for edge restricted tentacles. We can generalize Lemma~\ref{lem:tentaclemotion} to also hold for vertices using edge restricted tentacles. We claim this as a corollary.
\begin{corlemma}\label{cor:tentaclemotion}
Let \q\ move a distance $\delta$ to \q'\ on a line segment $s$ and let \rr\ move a distance $\epsilon$ to \rr'\ along a boundary edge $\be=[\v,\v']$, in such a way that the first segment of the tentacles from \q\ and \q'\ intersect the same reflex vertex, if the tentacle consists of multiple segments, and \cutb{\ZR(\q,\rr,\be)} and \cutb{\ZR(\q',\rr',\be)} have the same hiding vertex, then
$$
\len{\ZR(\q',\rr',\be)}  
=
\len{\ZR(\q,\rr,\be)} +
{\cal F}(\delta,\epsilon),
$$
such that 
\begin{align*}
{\cal F}(\delta,\epsilon) & = 
- F_0 + \sqrt{F_0^2 + F_1\delta + F_2\delta^2}
- F_3 + \frac{F_3+F_4\epsilon+F_5\delta+F_6\epsilon\delta+F_7\epsilon^2+F_8\epsilon^2\delta}{\sqrt{1+F_9\epsilon+F_{10}\epsilon^2+F_{11}\epsilon^3+F_{12}\epsilon^4}}
\nonumber\\
& \qquad
- F_{13} + \sqrt{\frac{F_{13}^2+F_{14}\epsilon+F_{15}\delta+F_{16}\epsilon\delta+F_{17}\epsilon^2+F_{18}\delta^2+F_{19}\epsilon^2\delta+F_{20}\epsilon\delta^2+F_{21}\epsilon^2\delta^2}{1+F_{22}\epsilon+F_{23}\epsilon^2}},
\end{align*}
where $F_0,\dots,F_{23}$ are constants.
\end{corlemma}
We will henceforth only work with edge restricted tentacles and just call them tentacles.

Given two points \q\ and \q'\ in \P, consider the tentacles \ZR(\q,\v,\be) and \ZR(\q',\v,\be), for each boundary edge \be\ and each vertex \v\ of \be. We define two sets \JFs(\q) and \JFs(\q') of tentacles such that $\ZR(\q,\v,\be)\in\JFs(\q)$ iff $\|\ZR(\q,\v,\be)\|\leq\|\ZR(\q',\v,\be)\|$ and $\ZR(\q',\v,\be)\in\JFs(\q')$ iff $\|\ZR(\q,\v,\be)\|>\|\ZR(\q',\v,\be)\|$. In this way, each end point \v\ of each boundary edge \be\ has exactly one tentacle in one of \JFs(\q) or~\JFs(\q').

From Lemma~\ref{lem:tentacle} it is clear that if both endpoints of a boundary edge $\be=[\v,\v']$ have tentacles in the same set, either \JFs(\q) or \JFs(\q'), then the tentacles in the set sees the whole edge \be. However, assume that $\ZR(\q,\v,\be)\in\JFs(\q)$ and $\ZR(\q',\v',\be)\in\JFs(\q')$, then $\|\ZR(\q,\v,\be)\|\leq\|\ZR(\q',\v,\be)\|$ and $\|\ZR(\q,\v',\be)\|>\|\ZR(\q',\v',\be)\|$. Since the length of a tentacle $\ZR(\q,\rr,\be)$ changes smoothly as \rr\ moves along \be\ from \v\ to \v'; see Corollary~\ref{cor:tentaclemotion}, there is some point \rs\ on \be\ such that $\|\ZR(\q,\rs,\be)\|=\|\ZR(\q',\rs,\be)\|$. If we include \ZR(\q,\rs,\be) into \JFs(\q) and \ZR(\q',\rs,\be) into \JFs(\q')\!, this guarantees that the tentacles in \JFs(\q) and \JFs(\q') together see \be.

Thus, if each edge \be\ either has the two tentacles of its endpoints in one set or there is a point \rs\ on \be\ such that $\ZR(\q,\rs,\be)\in\JFs(\q)$ and $\ZR(\q',\rs,\be)\in\JFs(\q')$, then the whole boundary is seen by the tentacles in the set. By considering the tour constructed by following each tentacle from the head to the tip and back in some order for each set \JFs(\q) and \JFs(\q'), by Lemma~\ref{lem:boundary}, the polygon \P\ is guarded by the tentacles in the two sets. We call each of the two sets \JFs(\q) and \JFs(\q') a {\em jellyfish\/} with {\em head\/} \q\ and \q', respectively, and $\JF(\q,\q')=\JFs(\q)\cup\JFs(\q')$ the {\em jellyfish pair\/} with heads \q\ and~\q'. We define \len{\JF(\q,\q')}, the {\em length\/} of a jellyfish pair, to be the length of its longest tentacle; see Figure~\ref{jellyfish}(a) and~(b).
\begin{figure*} 
			    \begin{center} \small
			    \input{\figdirjellyfish.pdf_t}

			    \parcaption{\capw}{\label{jellyfish}(a)~A pair of tentacles \protect\ZR(\protect\q,\protect\rr,\protect\be) and \protect\ZR(\protect\q',\protect\rr,\protect\be) in blue, %
(b)~a jellyfish pair~\protect\JF(\protect\q,\protect\q') in blue with jellyfish \protect\JFs(\protect\q) and \protect\JFs(\protect\q'), 
(c)~a minimum jellyfish pair~\protect\MJF(\protect\ea,\protect\eb) with minimum jellyfish~\protect\MJFs(\protect\ea) and~\protect\MJFs(\protect\eb) in blue.}
			    \end{center}
			    \end{figure*}

%

Let \ea\ and \eb\ be two extensions intersected by \Ta\ and \Tb\ respectively. These extensions exist by Lemma~\ref{lem:crossing} and we have the following lemma.
\begin{lemma}\label{lem:tentaclepair}
If \ua\ and \ub\ are intersection points of \Ta\ and \Tb\ with extensions \ea\ and \eb\ respectively, then $\len{\JF(\ua,\ub,)} \leq \maxlen{(\Ta,\Tb)}/2$. 
\end{lemma}
\begin{proof}
Without loss of generality, assume that a longest tentacle in \JF(\ua,\ub,) is \ZR(\ua,\rr,\be), for some point \rr\ on boundary edge \be. 
We distinguish five cases.
\begin{enumerate}
\item\label{case:bothtours}%
If \rr\ is an interior point of \be\ seen by \Ta, then follow a path from \ua\ along \Ta\ until we reach the first point \ua'\ on \Ta\ that sees \rr. Denote the subpath of \Ta\ thus constructed $\Ta[\ua,\ua']$ and let the other subpath $(\Ta\!\setminus\!\Ta[\ua,\ua'])\cup\{\ua,\ua'\}$ (the closure of the path) be denoted $\Ta[\ua',\ua]$. The path $\Ta[\ua',\ua]$ also sees \rr\ since \ua'\ sees \rr. The tentacle \ZR(\ua,\rr,\be) is at most as long as the shorter of the subpaths $\Ta[\ua,\ua']$ and $\Ta[\ua',\ua]$, whereby $\len{\ZR(\ua,\rr,\be)}\leq\len{\Ta}/2$ in this case.
\item
If \rr\ is an interior point of of \be\ but is not seen by \Ta, then it is seen by \Tb, and since $\len{\ZR(\ua,\rr,\be)}\leq\len{\ZR(\ub,\rr,\be)}$ and by the argument in Case~\ref{case:bothtours}.\@ using \ub\ and \Tb\ instead of \ua\ and \Ta, giving us $\len{\ZR(\ub,\rr,\be)}\leq\len{\Tb}/2$, we have $\len{\ZR(\ua,\rr,\be)}\leq\len{\Tb}/2$ in this case.
\item\label{case:opposite}%
If $\rr=\v$ is a vertex of \be\ seen by \Ta\ and \Ta\ also sees other points of \be, then by Lemma~\ref{lem:edgeconnected}, \Ta\ sees some connected set of \be\ that includes \v. Hence $\lim_{\be\ni\p\rightarrow\v}\|\ZT(\ua,\p)\|=\|\ZR(\ua,\v,\be)\|$ by the definition of edge restricted tentacle and using our argument from Case~\ref{case:bothtours} (including the fact that \Ta\ sees all points \p\ on \be\ in a connected neighbourhood of \v), we have $\len{\ZR(\ua,\v,\be)}\leq\len{\Ta}/2$ in this case.
\item
If $\rr=\v$ is a vertex of \be\ seen by \Ta\ but \Ta\ sees no other points of \be, then \Tb\ must see \v\ and other points of \be, otherwise there are points on \be\ not seen by any of \Ta\ and \Tb. Using the same argument as in Case~\ref{case:opposite} for \ub\ and \Tb, we have $\len{\ZR(\ub,\v,\be)}\leq\len{\Tb}/2$. Since $\len{\ZR(\ua,\v,\be)}\leq\len{\ZR(\ub,\v,\be)}$ we have $\len{\ZR(\ua,\v,\be)}\leq\len{\Tb}/2$ in this case.
\item
If $\rr=\v$ is a vertex of \be\ not seen by \Ta, then \be\ contains some connected set that includes \v\ and that is not seen by \Ta. This set must be seen by \Tb\ and we can use the argument in the previous step to establish that $\len{\ZR(\ua,\v,\be)}\leq\len{\Tb}/2$ also in this case.
\end{enumerate}
This concludes the proof.
\end{proof}%


Given two segment $s$ and $s'$ in \P\!, we define the {\em bases\/} along $s$ and $s'$ to be a pair of points having the property $\qs,\qsp=\argmin_{\q\in s,\q'\in s'}\big\{\len{\JF(\q,\q')}\big\}$, i.e., two points, \qs\ on $s$ and \qsp\ on $s'$, where \len{\JF(\qs,\qsp)} is minimal. We denote the jellyfish pair \JF(\qs,\qsp) by~\MJF(s,s')\!, the {\em minimum jellyfish pair}; see Figures~\ref{jellyfish}(b) and~(c).

From this definition and Lemma~\ref{lem:tentaclepair}, we have that
\begin{equation}\label{math:jellyfish}
\len{\MJF(\ea,\eb)} \leq \len{\JF(\ua,\ub)} \leq \maxlen{(\Ta,\Tb)}/2,
\end{equation}
where again, \ea\ and \eb\ are two extensions intersected by \Ta\ and \Tb, respectively.

Let \qa\ on \ea\ and \qb\ on \eb\ be the two bases of \MJF(\ea,\eb); see Figure~\ref{jellyfish}(c). We denote by \MJFs(\ea) the subset of tentacles having their heads at \qa\ and by \MJFs(\eb) the subset of tentacles having their heads at \qb. Each set of tentacles \MJFs(\ea) and \MJFs(\eb) is a {\em minimum jellyfish}. We call \qa\ the base of \MJFs(\ea) and \qb\ the base of \MJFs(\eb). Thus, $\MJF(\ea,\eb)=\MJFs(\ea)\cup\MJFs(\eb)$.

The main part of our algorithm, presented in the next section, computes, for every pair of extensions and every pair of boundary edges, the heads of the shortest tentacle pairs that see these edges. The heads are then used as potential bases and the length of the jellyfish pair with these heads as bases is computed. We keep the jellyfish pair with minimum length and its bases through the iteration and from the previous discussion we know that at the end of the iteration the final bases are~\qa\ and~\qb.


\section{The Algorithm}\label{sec:algorithm}

Our algorithm is illustrated in pseudo-code in Figure~\ref{alg:twr} and we show that it approximates a minmax two-watchman route by a factor of $\big(7\pi/6 + 3 - \sqrt{3} + \sqrt{5}\arcsin1/\sqrt{5}\big)$ and therefore by Inequality~(\ref{math:sumleq2max}) also a minsum two-watchman route by a factor twice as large.

The algorithm begins by running Belleville's algorithm~\cite{Bel:master,Bel:cover} to establish if the polygon is guardable by two point guards. If this is the case, it returns the two point guards computed by the algorithm. Note that if \P\ is two-guardable by point guards, our algorithm must obtain two such point guards to satisfy the approximation guarantee. Otherwise, it computes the set of extensions~\EE\ in $O(n\log n)$ time using a dynamic ray-shooting data structure~\cite{HerSur:rayshoot}, and initializes the solution to be a single shortest watchman route~\Ws~\cite{Tan:watchman,TanJia:touringpolygons} together with some arbitrary point in~\P. 

The rest of this section is devoted to showing how to implement Step~\ref{alg:twr:approx} of the algorithm. It consists of an iteration over all pairs of extensions. For each pair, we assume that each tour in a minmax two-watchman route intersects one of the extensions in the pair.
\begin{figure}
\begin{center}%
\fbox{\renewcommand{\algsize}{\small}
\begin{algorithm}{Two-Watchman-Route}
\INPUT{A simple polygon \P}
\OUTPUT{A two-watchman route $\WWt$ that sees~\P}
\STEP{\label{alg:twr:belleville}Run Belleville's algorithm to establish if the polygon is guardable by two point guards. If this is the case, return the two point guards computed by the algorithm}
\STEP{\label{alg:twr:extensions}Compute the set of extensions \EE\ in~\P}
\STEP{\label{alg:twr:watchman}Compute a shortest watchman route $\Ws$ in~\P}
\STEP{\label{alg:twr:init}Let $\WWSt\leftarrow(\Ws,{\rm arbitrary\ point\ in\ \P})$}
\STEP{\label{alg:twr:approx}%
\FOR{every pair of extensions $\ea,\eb\in\EE$, $\ea\neq\eb$}
\subSTEP{\label{alg:twr:approx:bases}
Establish the bases \qa\ and~\qb, and compute $\MJF(\ea,\eb)=\MJF(\qa,\qb)$
}
\subSTEP{\label{alg:twr:approx:reduce}Sort tentacles in \MJF(\ea,\eb) in decreasing order and remove covered tentacles $\rightarrow$ \MJR(\ea,\eb)}
\subSTEP{\label{alg:twr:approx:compute}Compute the two tours $\WWt=\big(\tour{\MJRs(\ea)},\tour{\MJRs(\eb)}\big)$}
\subSTEP{\label{alg:twr:approx:test}\shortIF{$\maxlen{\WWt}<\maxlen{\WWSt}$}{$\WWSt\leftarrow\WWt$}}
\ENDFOR
}
\STEP{\label{alg:twr:return}\RETURN{\WWSt}}
\end{algorithm}
}
\end{center}
\caption{\label{alg:twr}The Two-Watchman-Route algorithm.}
\end{figure}%


\subsection{Computing Tentacles and Bases}\label{sec:computation} 

The algorithm needs to find the two bases \qa\ on \ea\ and \qb\ on \eb. To this end, let \sei\ be the maximal line segments through \qi, orthogonal to \ei\ inside \P, for $i\in\{1,2\}$. The segment \sei\ partitions \P\ into two subpolygons \PL\ and \PR. The minimum jellyfish \MJFs(\ei) either has one tentacle that attains the length \len{\MJFs(\ei)} with its first segment orthogonal to \ei\ or at an endpoint of \ei, or it has two tentacles, one in \PL\ and the other in \PR, that attain this length. To prove this, note that if the single tentacle attaining the length \len{\MJFs(\ei)} does not have a first segment orthogonal to \ei, then by moving the head slightly along \ei, we can reduce the length of the tentacle and thereby the jellyfish. Similarly, if all tentacles attaining the maximal length are in the same subpolygon, say \PL, then we can move the head along \ei\ into \PL, again reducing the length of the jellyfish. In both cases, this contradicts that \MJFs(\ei) is a minimum jellyfish from a minimum jellyfish pair on~\ea\ and~\eb.
Thus, there are at most two longest tentacle pairs of \MJF(\ea,\eb), at least one pair of which attains the length \len{\MJF(\ea,\eb)}. 

We identify the different cases that can occur and design the data structures needed to compute the bases. These data structures are, for each reflex vertex in \P\ and each extension endpoint, the shortest path tree to every vertex in \P~\cite{GuiHerLevShaTar:visandshortpath}. The shortest path trees are augmented with the additional edges and leaves obtained by extending each tree edge until it intersects a boundary edge of \P\ without crossing any other tree edge. Also, for each augmented tree, a data structure is built, enabling us to find the common ancestor of any pair of nodes in the tree (vertices of \P) in constant time~\cite{HarTar:treeancestors}. These can be precomputed in linear time for each root vertex, i.e., in quadratic time in total.

\subsubsection*{Case~1}

If \MJF(\ea,\eb) has one unique longest tentacle, then we know from the previous discussion that it is a tentacle \ZR(\qei,\v,\be) for some vertex endpoint \v\ of some boundary edge \be\ and $i=1$ or $i=2$. The point \qei\ is the point on \ei\ that minimizes this distance and furthermore $\len{\ZR(\qei,\v,\be)}\leq\len{\ZR(\qet,\v,\be)}$. If \v\ is a reflex vertex, let \e\ be the extension associated to it and we have that $\ZR(\qei,\v,\be)=\SPb(\ei,\VP{\v}\cap\lhp{\e})$ and if \v\ is a convex vertex, $\ZR(\qei,\v,\be)=\SPb(\ei,\VP{\v})$. In both cases, \qei\ is the point of intersection between the shortest path and \ei. Given \ea\ and \eb\ we can, for each boundary edge \be\ and for each vertex end point \v, verify if \v\ is reflex or convex, compute \SPb(\ea,\VP{\v}\cap\lhp{\e}) and \SPb(\eb,\VP{\v}\cap\lhp{\e}) or \SPb(\ea,\VP{\v}) and \SPb(\eb,\VP{\v}), depending on the case, in linear time~\cite{ElgAvi:visibilitypolygon,GooOro:handbook,JoeSim:visibilitypolygon,Lee:visibilitypolygon,LeePre:shortpath,SacUrr:handbook}. We denote these tentacles \ZR(\ea,\v,\be) and \ZR(\eb,\v,\be) to indicate that the intersection points with \ea\ and \eb\ are not fixed given points. We select the shorter of the two, \ZR(\ei,\v,\be), identifying the intersection with the corresponding extension \ei, $i=1$ or $2$, keeping the \qei\ for which the tentacle is the longest. After iterating over all boundary edges, one potential base \qei\ remains on \ei\ and one \qet\ remains on~\et.
The whole process takes quadratic time and gives us the potential bases \qaa\ and~\qab\ on \ea\ and \eb, respectively.

\subsubsection*{Case~2}
If \MJF(\ea,\eb) has two longest tentacles, this can occur in three different ways. One possibility is that two tentacles \ZR(\ei,\v,\be) and \ZR(\ej,\v',\be'), $i\in\{1,2\}$, $j\in\{1,2\}$ and different boundary edges \be\ and \be', have exactly the same length. If $i=j$, the heads of \ZR(\ei,\v,\be) and \ZR(\ej,\v',\be') also coincide in this case, otherwise the base is found in Case~3. Such tentacles are discovered using the method described in the previous case.%
\begin{figure*} 
			    \begin{center} \small
			    \input{\figdirjellycomp.pdf_t}

			    \parcaption{\capw}{\label{jellycomp}Computing the bases of the minimum jellyfish pair.}
			    \end{center}
			    \end{figure*}

The second way is that two tentacles in \MJF(\ea,\eb) are \ZR(\ea,\rs,\be) and \ZR(\eb,\rs,\be), for some boundary edge \be, where \rs\ is a point on \be\ such that $\len{\ZR(\q_{\ea},\rs,\be)}=\len{\ZR(\q_{\eb},\rs,\be)}$, $\q_{\ea}$ and $\q_{\eb}$ being the points on \ea\ and \eb\ that make these tentacles as short as possible. We can, for each boundary edge $\be=[\v,\v']$, compute \ZR(\ea,\v,\be), \ZR(\eb,\v,\be), \ZR(\ea,\v',\be), and \ZR(\eb,\v',\be) as in Case~1. Now, if $\len{\ZR(\ea,\v,\be)}<\len{\ZR(\eb,\v,\be)}$ and $\len{\ZR(\ea,\v',\be)}>\len{\ZR(\eb,\v',\be)}$, or the inequalities are reversed, we let a point \rr\ slide along \be\ from \v\ to the other endpoint \v'; see Figure~\ref{jellycomp}(a). By Corollary~\ref{cor:tentaclemotion}, the lengths of \ZR(\ea,\rr,\be) and \ZR(\eb,\rr,\be) as \rr\ moves along \be\ are smooth functions (continuous and differentiable) except at positions denoted {\em event points\/} where any of the participating tentacles has one of the following properties:
\begin{enumerate}
\item
an interior point of the first edge of a tentacle intersects a vertex of~\P, or 
the first and second edge of the tentacle become collinear,
\item
an interior point of the last edge of a tentacle  intersects a vertex of~\P, or
the last and penultimate edge of the tentacle become collinear,
\item
the cut \cutb{\ZR(\ei,\rr,\be)}, $i\in\{1,2\}$, becomes intersects more than one vertex of~\P,
\item
the head or the tip of a tentacle either hits or leaves the boundary of~\P, and
\item
one tentacle goes from being shorter than the longer tentacle to becoming the longer tentacle, at which point the two tentacles have equal length.
\end{enumerate}
The positions on \be\ for the first three event point types can be obtained in linear time, if \ZR(\ei,\rr,\be) consists of at least two segments having as common endpoint the reflex boundary vertex \ui, by a traversal of the augmented shortest path tree to each vertex in \P\ from the root \ui, giving us a total of $O(n)$ such event points. If \ZR(\ei,\rr,\be) consists of just one segment, the positions for the first three event point types on \be\ can be obtained in linear time by a traversal of the augmented shortest path trees rooted at the two end points of~\ei.

The fourth type is obtained by finding the point at which the first edge of the tentacle hits either endpoint of \ei\ and where the last edge of the tentacle intersects \cutb{\ZR(\ei,\rr,\be)} orthogonally on the boundary which can happen only once for each tentacle in each interval between event points of the first three types. Finally, the last type of event points can occur only a constant number of times in each interval between event points of the first four types, since here each function is the sum of at most three square roots of rational polynomials of constant degree and two such functions can be equal in at most a constant number of points. The number of event points between \v\ and \v'\ on \be\ is at most linear and iterating over all possible boundary edges, takes quadratic time to solve this case. This gives us the potential bases \qba\ and~\qbb\ on \ea\ and \eb, respectively.

\subsubsection*{Case~3}
In the third possibility that \MJF(\ea,\eb) has two longest tentacles, the two tentacles are \ZR(\q,\va,\bea) and \ZR(\q,\vb,\beb) for two boundary edges \bea\ and \beb\ with vertex endpoints \va\ and \vb\ respectively and \q\ on \ei\ such that $\len{\ZR(\q,\va,\bea)}=\len{\ZR(\q,\vb,\beb)}$. We compute \ZR(\ei,\va,\bea) and \ZR(\ei,\vb,\beb) and let \qei\ and \qei'\ be the heads of \ZR(\ei,\va,\bea) and \ZR(\ei,\vb,\beb), respectively, on \ei. We let a point \q\ slide along \ei\ from \qei\ to \qei'. Again, by Corollary~\ref{cor:tentaclemotion}, the lengths of \ZR(\q,\va,\bea) and \ZR(\q,\vb,\beb) as \q\ moves along \ei\ are smooth functions (continuous and differentiable) except at the event points established in Case~2. At these event points, the structures of the tentacles are updated and we can test whether the two tentacles have equal length for some head point on \ei\ before the next event point. The number of event points between \qei\ and \qei'\ on \ei\ is at most linear. We perform these steps also on \et\ and take the shortest pair as the representative base for the pair $(\va,\vb)$ of vertices. Iterating over all possible pairs of boundary edges and their endpoints, the process takes cubic time in total and gives us the potential bases \qca\ and~\qcb\ on \ea\ and \eb, respectively.

\subsubsection*{Case~4}
If \MJF(\ea,\eb) has three longest tentacles, this can occur in two ways. The first possibility if some combination of three tentacles with common heads from the three previous cases exist, in which case these can be established using the previously described methods.

The second possibility that \MJF(\ea,\eb) has three longest tentacles occurs if there are two boundary edges $\bea=[\va,\vap]$ and $\beb=[\vb,\vbp]$\ such that $\len{\ZR(\qss,\rs,\bea)}=\len{\ZR(\et,\rs,\bea)}=\len{\ZR(\qss,\vb,\beb)}$, with \qss\ on \ei; see Figure~\ref{jellycomp}(b). We can establish the event points by a combination of the methods in Cases~2 and~3. Begin by finding the interval $[\qei,\qei']$ on \ei\ by computing the tentacles \ZR(\ei,\va,\bea) and \ZR(\ei,\vb,\beb) having the heads \qei\ and \qei', and then the event points generated by \ZR(\q,\v,\be) and \ZR(\q,\v',\be') in order as \q\ moves along \ei\ from \qei\ to \qei'. Subsequently, establish the event points generated by \ZR(\ei,\rr,\be) and \ZR(\et,\rr,\be) in order as \rr\ moves along \be. For each pair of intervals between event points along \ei\ and along \bea, we verify if we can establish a point \q\ in the interval on \ei\ and a point \rr\ in the interval on \bea\ such that \ZR(\q,\rr,\bea), \ZR(\q,\vb,\beb), and \ZR(\et,\rr,\bea) have equal length by solving the system of equalities given by the length functions of the three tentacles. Within each interval the length functions are smooth according to Corollary~\ref{cor:tentaclemotion} so this takes constant time. In the worst case, we have to consider a linear number of intervals each on \ei\ and \bea\ taking quadratic time. Iterating over all possible pairs of boundary edges and their endpoints, the computation takes quartic time in total and gives us the potential bases \qda\ and~\qdb\ on \ea\ and \eb, respectively.

\subsubsection*{Case~5}
If \MJF(\ea,\eb) has four longest tentacles, this can occur in two ways. Again, the first possibility is if some combination of four tentacles with common heads from the four previous cases exist, in which case these can be established using the previously described methods. 

The second possibility that \MJF(\ea,\eb) has four longest tentacles occurs if there are two boundary edges \bea\ and \beb\ such that $\len{\ZR(\qsa,\rsa,\bea)}=\len{\ZR(\qsb,\rsa,\bea)}=\len{\ZR(\qsa,\rsb,\beb)}=\len{\ZR(\qsb,\rsb,\beb)}$, with \qsa\ on \ea, \qsb\ on \eb, \rsa\ on \bea, and \rsb\ on \beb; see Figure~\ref{jellycomp}(c). We can establish the event points by extending the method in Case~4, giving a linear number of intervals on each segment \ea, \eb, \bea, and \beb, where the length functions of the tentacles are smooth according to Corollary~\ref{cor:tentaclemotion}. We verify if we can establish point \qa, \qb, \ra, and \rb\ in their respective interval such that \ZR(\qa,\ra,\bea), \ZR(\qb,\ra,\bea), \ZR(\qa,\rb,\beb), and \ZR(\qb,\rb,\beb) have equal length by solving the system of equalities given by the length functions of the four tentacles which takes constant time in each interval. In the worst case, we have to consider a linear number of intervals each on \ea, \eb, \bea, and \beb, thus taking quartic time. Iterating over all possible pairs of boundary edges and their endpoints, the computation takes $O(n^6)$ time in total and gives us the potential bases \qea\ and~\qeb\ on \ea\ and \eb, respectively.

\subsubsection*{Case~6}
If \MJF(\ea,\eb) has five or more longest tentacles, they must occur as some combination of tentacles structured as in the previous cases and they can therefore be obtained using the methods described above.

\bigskip\noindent
This gives us five pairs of potential bases, $\big(\qaa\!,\qab\big),\big(\qba\!,\qbb\big),\ldots,\big(\qea\!,\qeb\big)$, for which we can compute the jellyfish pairs $\JF(\qaa\!,\qab),\JF(\qba\!,\qbb),\ldots,\JF(\qea\!,\qeb)$, each in quadratic time, and from these we select the minimum one. By Corollary~\ref{cor:tentaclemotion} and the previous discussion, it follows that this jellyfish pair is a minimum jellyfish pair \MJF(\ea,\eb) with bases \qa\ and \qb\ on \ea\ and \eb, respectively. 

In this way, we have an $O(n^6)$ time subroutine to find the bases \qa\ and \qb\ on \ea\ and \eb, respectively. Given the bases, the computation of \MJF(\ea,\eb) takes an additional quadratic time. The whole process of Step~\ref{alg:twr:approx:bases} of the algorithm thus takes $O(n^6)$~time.

\subsection{Establishing the Tours}\label{sec:tours}

Given a minimum jellyfish pair \MJF(\ea,\eb), we sort the tentacles in decreasing order of length and for each tentacle \ZR(\qi,\rr,\be) in decreasing order check if \cutb{\ZR(\qi,\rr,\be)} is already covered by a longer tentacle. If so, \ZR(\qi,\rr,\be) is removed from the jellyfish pair. We call the resulting jellyfish pair {\em reduced\/} and denote it by~\MJR(\ea,\eb) with individual jellyfish \MJRs(\ea) and~\MJRs(\eb).

Let \tour{{\cal J}} denote the two relative convex hulls of each of the jellyfish in the jellyfish pair ${\cal J}$ inside \P\!. Without loss of generality, if ${\cal J}={\cal J}_1\cup{\cal J}_2$ with 
${\cal J}_1$ and ${\cal J}_2$ being the two jellyfish, \tour{{\cal J}_1} and \tour{{\cal J}_2} also denote the individual relative convex hulls of ${\cal J}_1$ and~${\cal J}_2$ in~\P.

Let \Wa\ and \Wb\ be the relative convex hulls \tour{\MJRs(\ea)} and \tour{\MJRs(\eb)} in \P, respectively. These tours can be computed in linear time by first following the shortest path from each tentacle tip to the next, cyclically around the corresponding head of each jellyfish and then applying the algorithm by Toussaint~\cite{Tou:relconvexhull}; see Figure~\ref{approxfig}(a).
\begin{figure*} 
			    \begin{center} \small
			    \input{\figdirapproxfig.pdf_t}

			    \parcaption{\capw}{\label{approxfig}Illustrating the construction of the tour and the partition into subpaths used in the proof of Lemma~\protect\ref{lem:approx}.}
			    \end{center}
			    \end{figure*}

Consider a polygonal tour that has reflex vertices only at reflex vertices of~\P. A maximal consecutive subsequence of edges of the tour is called a {\em reflex chain\/} of the tour, if each interior vertex of the chain is reflex in the tour. The end vertices of a reflex chain  must therefore be convex vertices of the tour. In contrast, a {\em convex chain\/} of the tour is a maximal consecutive subsequence of edges such that each interior and end vertex of the chain is convex in the tour. We note that a single convex vertex of a tour is a convex chain if the preceding and subsequent vertices are reflex.

\begin{lemma}\label{lem:approx}
The tours $(\Wa,\Wb)$ obtained by algorithm {\it Two-Watchman-Route\/} form a two-watchman route and 
$
\maxlen{(\Wa,\Wb)} \leq (7\pi/6 + 3 - \sqrt{3} + \sqrt{5}\arcsin1/\sqrt{5})\maxlen{(\Ta,\Tb)} \approx 5.969\maxlen{(\Ta,\Tb)}
$
\end{lemma}
\begin{proof}
The correctness of the algorithm follows from Lemmas~\ref{lem:boundary},~\ref{lem:tentacle}, and the fact that since the two tours together see every boundary edge this ensures that they form a two-watchman route.

To prove the approximation bound, we assume that \Ta\ and \Tb\ do not intersect, otherwise $\len{\Ws}\leq\len{\Ta}+\len{\Tb}\leq2\maxlen{(\Ta,\Tb)}$ immediately proving the lemma since $\maxlen{(\Wa,\Wb)}\leq\len{\Ws}$ and the algorithm initializes the two-watchman route with a shortest watchman tour and a point in Step~\ref{alg:twr:watchman} and then only updates its two-watchman route in Step~\ref{alg:twr:approx:test} if its max length is strictly smaller than that of the current route pair. 

The algorithm computes the reduced minimum jellyfish pair \MJR(\ea,\eb) in Step~\ref{alg:twr:approx:reduce}. By trying all pairs of extensions in Step~\ref{alg:twr:approx}, the algorithm must necessarily consider a pair, \ea\ and \eb, intersected by the tours \Ta\ and \Tb; see Lemma~\ref{lem:crossing}. 
Consider the tentacles in \MJRs(\ea) and \MJRs(\eb) centered on the bases \qa\ on \ea\ and \qb\ on \eb, respectively. 
Every tentacle has length at most $R=\maxlen{(\Ta,\Tb)}/2$ by Inequality~(\ref{math:jellyfish}), since we can assume that \Ta\ intersects \ea\ and \Tb\ intersects \eb, whereby the geodesic radii of \Wa\ and \Wb\ are both at most~$R$. 

Each convex chain of \Wa\ has length at most $2\pi R$, where $R$ is an upper bound on the length of each geodesic shortest path from \qa\ to any point on \Wa, since the circle is the longest convex curve of radius~$R$. This follows from Archimedes' axioms for arc-length~\cite{Arc:arc-length,Bro:archimedes}.

We make a case analysis and bound the length of \Wa\ for each case separately. 
\begin{description}
\item[\Wa\ has one convex chain.]
The convex chain of \Wa\ has length at most $2\pi R$ as we noted above. 
The length of the possible reflex chain of \Wa\ is bounded by at most two radii from \qa\ to the circle perimeter since its length is bounded by that of the two tentacles of \MJFs(\ea) connecting to the endpoints of the reflex chain. Hence, 
\begin{equation}\label{eqn:caseA}
\len{\Wa} \leq 2\pi R + 2R = (\pi+1)\maxlen{(\Ta,\Tb)}.
\end{equation}

\item[\Wa\ has at least two convex chains.]
This case is further subdivided into the cases:
\begin{description}
\item[\Wa\ or its interior intersects both \Ta\ and~\Tb.]
If \Wa\ intersects \Ta, let \pa\ be an intersection point of \Wa\ with \Ta. From \pa\ move counterclockwise along \Wa\ until the endpoint of a tentacle from \MJRs(\ea) is reached at \paL. Similarly, move clockwise along \Wa\ until the endpoint of a tentacle from \MJRs(\ea) is reached at \paR. Since the region bounded by moving counterclockwise from \paR\ along \Wa\ to \paL, from \paL\ along the tentacle of \MJRs(\ea) to \qa, and from \qa\ along a tentacle of \MJRs(\ea) to \paR\ forms a pseudotriangle, relatively convex inside \P, the shortest path from \pa\ to \qa\ has length at most that of one of the tentacles of \MJRs(\ea), which is bounded by~$R$. If \Wa\ does not intersect \Ta, then \Ta\ lies interior to \Wa\ and we let \pa\ be a point closest to \qa. Again, the length of shortest path from \pa\ to \qa\ is bounded by $R$, since extending the first edge until it hits the tour \Wa\ at some point \pa', the shortest path from \pa'\ to \qa\ contains \pa\ and since \pa'\ lies on \Wa, we can use the previous argument to show that the length of the shortest path from \pa'\ to \qa\ is bounded by $R$.

If \Wa\ intersects \Tb, let \pb\ be intersection point of \Wa\ with \Tb, otherwise let \pb\ be a point on \Tb\ closest to \qa. In the same way as above, we can bound the length of the shortest path from \pb\ to \qa\ by $R$.

Now, construct a tour \Wbnd\ by following \Ta\ from \pa\ around in counterclockwise order, following the shortest path from \pa\ to \qa, the shortest path from \qa\ to \pb, from \pb\ around \Tb\ in counterclockwise order, and then back along the shortest paths via \qa\ to \pa. It is clear that \Wbnd\ is a watchman tour and therefore has length no shorter than \Ws\ and since the algorithm initializes the two-watchman route with \Ws\ and an arbitrary point in Step~\ref{alg:twr:watchman} and then only updates its two-watchman route in Step~\ref{alg:twr:approx:test} if it has max-length strictly smaller than the current route pair, the length of \Wa\ is bounded by
\begin{equation}\label{eqn:caseB}
\len{\Wa} 
<
\len{\Ws} 
\leq
\len{\Wbnd} 
\leq
\len{\Ta} + \len{\Tb} + 4R
\leq
4\maxlen{(\Ta,\Tb)}.
\end{equation}

\item[\Wa\ or its interior intersects at most one of \Ta\ or~\Tb.]
Without loss of generality, we assume that \Ta\ does not intersect \Wa\ or the interior of \Wa, otherwise we interchange the roles of \Ta\ and \Tb\ in the argument below.

Cut \P\ along the segments of \Wa\ into disjoint pieces, thus partitioning \P\ into separate components; see Figure~\ref{approxfig}(b). Let \QQ\ be the component containing \Ta. The subpath $\Pi$ of \Wa\ bounding \QQ\ is either a single segment connecting two reflex vertices of \Wa\ or it consists of a single convex chain \C\ between two segments having reflex vertices \v\ and \v'\ at the endpoints. 
Extend $\Pi$ at both ends along \Wa\ until convex vertices are reached at both ends, these must exist since \Wa\ has at least two convex chains. Let the path thus obtained be $\Pi'$\!. 
If $\Pi$ includes a convex chain, consider the two paths $\Pi'\setminus\C$ that we denote by \Pil\ and \Pir, where \Pil\ appears before \Pir\ during a counterclockwise traversal of \Wa\ starting at a point in \C; see Figure~\ref{approxfig}(b). 

Walk along \C\ from each endpoint \v\ and \v'\ until convex vertices \vbar\ and \vbar'\ are reached for which the associated tentacle cuts are not covered by \Tb, if there are any. These cuts must be covered by \Ta. Let \Ca\ be the subpath of \C\ from \vbar\ to \vbar'. If $\Ca=\C$, then already the end points of \C\ have associated tentacle cuts that are not covered by~\Tb\ and if all associated tentacle cuts to the vertices of \C\ are covered by \Tb, then we let $\Ca=\emptyset$.

If \Ca\ is not empty, consider the two endpoints \vbar\ and \vbar'\ of \Ca\ (if the endpoints coincide, \Ca\ is a single point and has length~$0$, so we assume that the endpoints do not coincide); see Figure~\ref{approxfig}(b). Let \e\ and \e'\ be the two tentacle cuts covered by the two endpoints \vbar\ and \vbar'\ of \Ca. Since \Tb\ does not cover any of \e\ or \e'\!, \Ta\ must have points to the left of these cuts. Without loss of generality, we assume that \vbar\ is passed before \vbar'\ during a counterclockwise traversal of \Wa\ starting at a point interior to~\Ca. We make the additional observation that any tentacle cut associated to convex vertices of $\Wa\setminus\Ca$ must be covered by \Tb, otherwise \Ta\ would intersect \Wa\ or have points in its interior.%
\begin{figure*} 
			    \begin{center} \small
			    \input{\figdirboundfig.pdf_t}

			    \parcaption{\capw}{\label{boundfig}Illustrating the proof of Lemma~\protect\ref{lem:approx}. Bounding the length of~\protect\Ca.}
			    \end{center}
			    \end{figure*}

Let $l$ and $l'$ be the lines through \e\ and \e'\ and let \p\ be the intersection point between $l$ and $l'$, assuming that it exists. If $l$ and $l'$ are parallel or \p\ lies after \vbar\ on $l$ in the direction of \e, then the angle between the tentacle of \MJRs(\ea) connecting \qa\ and \vbar\ and the tentacle of \MJRs(\ea) connecting \qa\ and \vbar'\ containing \Ca\ is at most $\pi$, since \qa, \vbar', \p, and \vbar\ form a convex quadrilateral (\p\ can be taken to be a point on $l$ implicitly at infinity, if $l$ and $l'$ are parallel); see Figure~\ref{boundfig}(a). Each segment of \Ca\ can be projected onto a half circle centered at \qa\ having radius $R$ without overlap so the length of \Ca\ is therefore bounded by~$\pi R$ in this case. 

Now, if \p\ lies before \vbar\ on $l$ in the direction of \e; see Figure~\ref{boundfig}(b), then since \Ta\ lies in \QQ\ without intersecting \Ca, \Ta\ must intersect the two cuts \e\ and \e'. Assume these intersection points are \rr\ and \rr'. The distance between \rr\ and \rr'\ is at most $R$ since $\len{\Ta}\leq\maxlen{(\Ta,\Tb)}=2R$ by Inequality~\ref{math:jellyfish}. Assume that the longest tentacle of \MJRs(\ea) connecting \qa\ with a convex vertex of \Ca\ has length $D\leq R$ and let \dd\ and \dd'\ be the points at distance exactly $D$ from \qa\ to $l$ and $l'$, respectively. The point \dd\ must lie between \vbar\ and \rr\ on \e, otherwise \Ta\ intersects the longest tentacle of \MJRs(\ea) and therefore also \Ca. For the same reason, the point \dd'\ must lie between \vbar'\ and \rr'\ on \e'.
The distance $\len{\dd,\dd'}\leq R$ since \Ta\ does not intersect \Wa. Given that the angle at \dd, \qa, \vbar, and the angle at \dd', \qa, \vbar', are each at most $\pi/2$, that the angle at \dd, \qa, \dd'\ is $\theta$, the segments of \Ca\ can be projected onto a semicircle centered at \qa\ having radius $D$ without overlap. The length of this semicircle is $(\pi+\theta)D$, which is maximized for $\theta=\pi/3$ and~$D=R$ as obtained by standard analytic methods. Thus, 
\begin{align}\label{eqn:semicircle}
\len{\Ca}
&\leq
4\pi R/3
\end{align}
is the maximum bound in all cases.

Next, we bound the length of the remainder of \Wa. Let \p\ now be an intersection point between \Wa\ and \ea. From \p, walk counterclockwise along \Wa\ back to \p, shortcutting those convex vertices associated to tentacle cuts only covered by \Ta\ and let \Wpa\ be the tour thus obtained. Let \MJWs(\ea) be the subset of tentacles of \MJRs(\ea) for which the associated tentacle cut is covered by \Tb. By construction, the tour \Wpa\ is the relative convex hull of the tentacles of \MJWs(\ea) in \P\!. All convex vertices on \Wa\ associated to tentacle cuts covered only by \Ta\ lie on \Ca\ by our observation above, giving us
\begin{align}\label{eqn:nonintersection}
\len{\Wa}
&\leq
\len{\Ca}+2R+\len{\Wpa},
\end{align}
since the union of the two tentacles connecting \qa\ with the endpoints of \Ca\ intersects \Wpa, they together cover the tentacle cuts covered by \Wa\ at the same points, and \Wa\ is relatively convex. Our objective now is to bound the length of~\Wpa.

Let \e\ be an extension covered by \Ta\ such that \Ta\ makes a reflection on \e. Such an extension must exist by Lemma~\ref{lem:crossing} since, if \Ta\ properly covers \e, then it must cover some other extension in $\lhp{\e}$. Without loss of generality, we can therefore assume that $\ea=\e$ and thus that all but one point of \Ta\ lie in $\P\setminus\lhp{\ea}$. It also follows that \Tb\ cannot cover \ea, since this would mean that \Ta\ could be made shorter contradicting that \Ta\ and \Tb\ are as short as possible. Hence, all points of \Tb\ lie in~$\P\setminus\lhp{\ea}$. 

We identify two different cases.
\begin{description}
\item[\Wpa\ has at least two reflex chains or one reflex chain with both endpoints in~$\P\setminus\lhp{\ea}$]

We prove that \Tb\ must intersect \Wpa. Let \e'\ and \e''\ be the two tentacle cuts associated to convex vertices of \Wpa\ on either side of a reflex chain with both end points in~$\P\setminus\lhp{\ea}$. The cuts \e'\ and \e''\ must exist and do not intersect in $\P\setminus\lhp{\ea}$, otherwise \Wpa\ does not have the stated reflex chain. Now, \Tb\ must cover \e'\ and \e''\ so it either covers \e'\ exactly at the reflection point of \Wpa\ and \e'\ or \Tb\ covers \e''\ exactly at the reflection point of \Wpa\ and \e''\!, giving us an intersection point between \Wpa\ and \Tb\ in either case, or \Tb\ contains a path that connects a point in $\big(\P\setminus\lhp{\ea}\big)\cap\lhp{\e'}$ with a point in $\big(\P\setminus\lhp{\ea}\big)\cap\lhp{\e''}$ which must intersect \Wpa\ since \e'\ and \e''\ are separated by reflex vertices on \Wpa\ touching the boundary of~\P\ and \Tb\ has no point in the interior of~\lhp{\ea}; see Figure~\ref{coverfig}(a).
\begin{figure*} 
			    \begin{center} \small
			    \input{\figdircoverfig.pdf_t}

			    \parcaption{\capw}{\label{coverfig}Illustrating the proof of Lemma~\ref{lem:approx}. (a) Tours \protect\Tb\ and \protect\Wpa\ intersect if \protect\Wpa\ has reflex vertices. (b)--(e) Constructing the chain \protect\Cbi\ from~\protect\Ci.}
			    \end{center}
			    \end{figure*}

We can construct a tour that connects \qa\ with the closest point on \Tb, follows \Tb\ around in counterclockwise order and then connects to \qa. This tour has the following properties, that we call {\em bounding properties},
\begin{enumerate}
\item it covers all the tentacle cuts covered by \Wpa\ (and thus by \MJWs(\ea)), 
\item it lies completely in $\ea\cup(\P\setminus\lhp{\ea})$,
\item it has a convex vertex at \qa, and 
\item it has length at most $2R+\len{\Tb}$.
\end{enumerate}
The last bounding property follows since \Tb\ intersects \Wpa\ (or its interior).

Now, let \Ub\ be the shortest tour that obeys the bounding properties. We show how to extend \Ub\ to a new tour \Vb\ so that \Wpa\ is contained in \Vb\ and $\len{\Vb}\leq\sqrt{2}\len{\Ub}$. Since \Wpa\ is relatively convex and contained in \Vb, we immediately have 
\begin{align}\label{eqn:opttour}
\len{\Wpa}
&\leq
\len{\Vb}
\leq
\sqrt{2}\len{\Ub}
\leq
2\sqrt{2}R+\sqrt{2}\len{\Tb}.
\end{align}%

To prove Inequality~(\ref{eqn:opttour}), follow \Ub\ counterclockwise from \qa, subdividing the tour into convex chains \Ci\!, $1\leq i\leq k$, such that \uia\ is the reflex vertex before \Ci\ on \Ub, and \uib\ is the reflex vertex after \Ci\ on~\Ub. In the cases where the first vertex after \qa\ or before \qa\ is convex, we consider \qa\ to be the vertex \uaa\ and \ukb\ according to the case. 
For each convex chain \Ci\!, let \mi\ be the number of convex vertices on the chain, denoted by $\via,\ldots,\vimi$, these are the tips of the tentacles of \MJWs(\ea), and let $\eia,\ldots,\eimi$ be the corresponding associated tentacle cuts; see Figure~\ref{coverfig}(b)--(e). It is clear that no tentacle of \MJWs(\ea) intersects the interior of \lhp{\eij}, for any tentacle cut \eij, since $\MJWs(\ea)\subseteq\MJRs(\ea)$ and if it did, it would cover \eij\ whereby \eij\ and its associated tentacle would be discarded as \MJRs(\ea) is constructed. (To guarantee this property is in fact the reason why the algorithm constructs \MJR(\ea,\eb) in Step~\ref{alg:twr:approx:reduce}.)

To obtain \Vb\ from \Ub, we replace each convex chain \Ci, from \uia\ to \uib\ by a different path \Cbi. 
If $\mi=1$, i.e., \Ci\ has only one convex vertex \via, then let \riaa\ be the endpoint of \SP(\uia\!,\eia) on the tentacle cut \eia\!. Similarly, let \riab\ be the endpoint of \SP(\uib\!,\eia) on \eia\!. We let \Cbi\ be the shortest path from \uia\ to \eia, the path from \riaa\ to \riab\ on \eia, followed by the shortest path from \eia\ to \uib; see Figure~\ref{coverfig}(b). The convex vertex \via\ lies on the segment from \riaa\ to \riab\ on \eia, otherwise \Ub\ can be made shorter contradicting that \Ub\ is the shortest tour obeying the bounding properties. The length of the subpath of \Cbi\ from \uia\ to \via\ is upper bounded by the length of the two catheti connecting \uia\ with \via, one being parallel to \eia\ and the other orthogonal to \eia. Similarly, the length of the subpath of \Ci\ from \uia\ to \via\ is lower bounded by the length of the line segment from \uia\ to \via\ forming the hypotenuse of a right angled triangle with \uia\ and \via\ as two corners. From elementary geometry it follows that the sum of the lengths of the catheti is bounded by $\sqrt{2}$ times the length of the hypotenuse in a right angled triangle. Thus, the length of the subpath of \Cbi\ from \uia\ to \via\ is at most $\sqrt{2}$ times the length of the subpath of \Ci\ from \uia\ to \via.
The same argument ensures that the length of the subpath of \Cbi\ from \via\ to \uib\ is also at most $\sqrt{2}$ times the length of the subpath of \Ci\ from \via\ to~\uib.

We proceed inductively on the parameter \mi\!, the number of tentacle cuts covered along the path \Ci\!, as follows. Let \pia\ be the intersection of \eia\ and \eib\ (or the directed lines along \eia\ and \eib). The point \pia\ must exist, otherwise \Ci\ is not a convex chain but contains reflex vertices. Let \aia\ be the angle of the wedge $\P\setminus\big(\lhp{\eia}\cap\lhp{\eib}\big)$ at~\pia; see Figures~\ref{coverfig}(c)--(e).

We let $\Cbi\defeq\Cbiam=\Cbiab\cup\Cbibm$ be the path from \uia\ to \uib, where \Cbiab\ is a path from \uia\ to \vib\ that covers \eia, contains \via\!, and is at most $\sqrt{2}$ times longer than the path from \uia\ to \vib\ along \Ci\!. Since \Ci\ from \vib\ to \uib\ covers fewer than \mi\ tentacle cuts, we can assume that \Cbibm\ from \vib\ to \uib\ constructed inductively covers the remaining tentacle cuts, contains the vertices $\vib,\ldots,\vimi$, and is at most $\sqrt{2}$ times longer than the path from \vib\ to \uib\ along~\Ci\!.

\medskip
It remains to describe the construction of \Cbiab\ from \uia\ to \vib\!, that covers one tentacle cut \eia\!, passes one vertex \via\!, and is at most $\sqrt{2}$ times longer than the path from \uia\ to \vib\ along \Ci\!. 

\begin{description}
\item[]
If \aia\ is obtuse or a right angle, let \SP(\eia,\eib) be the shortest path between \eia\ end \eib\!. The path either degenerates into the single point \pia\!, if \eia\ end \eib\ intersect, or it is a path connecting a point \riab\ on \eia\ to a point \riba\ on \eib\!. In this case, \Cbiab\ is the shortest path from \uia\ to \eia\!, the path from \riaa\ to \riab\ (or \pia) on \eia\!, followed by the shortest path from \eia\ to \eib\!, the path from \riba\ (or \pia) to \vib; see Figure~\ref{coverfig}(c) and~(d).

As previously noticed, the length of the subpath of \Cbiab\ from \uia\ to \via\ is upper bounded by the length of the two catheti connecting \uia\ with \via\!, one being parallel to \eia\ and the other orthogonal to \eia\!. Similarly, the length of the subpath of \Ci\ from \uia\ to \via\ is lower bounded by the length of the line segment from \uia\ to \via\ forming the hypotenuse of a right angled triangle with \uia\ and \via\ as two corners. From elementary geometry it follows that the sum of the lengths of the catheti is bounded by $\sqrt{2}$ times the length of the hypotenuse in a right angled triangle. Thus, the length of the subpath of \Cbiab\ from \uia\ to \via\ is at most $\sqrt{2}$ times the length of the subpath of \Ci\ from \uia\ to \via\!.
The same argument ensures that the length of the subpath of \Cbiab\ from \via\ to \vib\ is also at most $\sqrt{2}$ times the length of the subpath of \Ci\ from \via\ to~\vib\!, since the angle \aia\ is obtuse or a right angle.

\item[]
If \aia\ is acute, consider the segment $\s\defeq[\qa,\pia]$ allowed to pass through the boundary of \P\!, if \eia\ and \eib\ do not intersect. Let  \ls\ be the line orthogonal to \s\ passing through \pia\!. We slide \ls\ along the segment \s\ from \pia\ towards \qa\ until \ls\ touches a point of \Ci\ between \via\ and \vib\ inclusive, giving the line \lss\!. In fact, this intersection point is one of \via\ or \vib\!. Let \riab\ be the intersection point of \SP(\eia\!,\lss) on \eia\ and \rias\ the intersection point of \SP(\eia\!,\lss) on \lss\!. 
Similarly, let \riba\ be the intersection point of \SP(\eib\!,\lss) on \eib\ and \ribs\ the intersection point of \SP(\eib\!,\lss) on~\lss\!. 
As before, \riaa\ is the intersection point of \SP(\uia\!,\eia) with \eia.
Now, \Cbiab\ is the shortest path from \uia\ to \eia\!, the path from \riaa\ to \riab, followed by the shortest path from \eia\ to \lss\ at \rias, the path from \rias\ to \ribs\ on \lss\!, the shortest path from \lss\ to \eib\ ending at \riba\!, followed by the path from \riba\ to \vib; see Figure~\ref{coverfig}(e). Since \aia\ is acute, the interior angles between \eia\ and \lss\ and between \eib\ and \lss\ are both obtuse or right angles. Hence, the lengths of the subpaths of \Cbi\ from \uia\ to \via\ and from \via\ to \vib\ are upper bounded by the lengths of the corresponding catheti of the right angled triangles, one connecting \uia\ and \via\!, with one catheter being parallel to \eia\!, and one connecting \via\ to \vib\!, with one catheter parallel to \lss\!. As before, the length of this subpath of \Cbiab\ is at most $\sqrt{2}$ times the length of the subpath of \Ci\ from \uia\ to~\vib\!.

\end{description}

Thus, we have proved Inequality~(\ref{eqn:opttour}) for both cases.

\item[\Wpa\ has no reflex chain or one reflex chain with at most one endpoint in~$\P\setminus\lhp{\ea}$]

If there is a point of \Tb\ at distance at most $R$ from \qa, then there exists a tour having the bounding properties as defined above and we can bound the length of \Wpa\ in exactly the same way as in the previous case. Hence, from now on, we assume that all points of \Tb\ have distance greater than $R$ from~\qa.

\medskip
Since each tentacle of \MJWs(\ea) has length at most $R$, every point of the tour \Wpa\ has distance at most $R$ to \qa\ and each tentacle is perpendicular to the tentacle cut which in turn intersects \Tb, except where \Wpa\ has a reflex chain. 
Let \Da\ be the possible set of points that can be points of \Wpa, given \Tb\ but for any possible set of tentacle cuts. Thus, \Wpa\ is contained in \Da. The perimeter of \Da\ consists of connected circular arcs with possible line segments in between. If \Wpa\ has a reflex chain, the perimeter alternates between circular arcs centered at \qa\ and straight line segments connecting those reflex boundary vertices where \Wpa\ also has reflex vertices. Any other polygon boundary part intersecting the interior of the region is disregarded as we want to obtain the maximum possible perimeter length in relation to $R$, given that \Wpa\ has at most one reflex chain. 
\Da\ is the region painted light green in Figure~\ref{circfig}(a).%
\begin{figure*} 
			    \begin{center} \small
			    \input{\figdircircfig.pdf_t}

			    \parcaption{\capw}{\label{circfig}Illustrating the proof of Lemma~\protect\ref{lem:approx}. Bounding the length of~\protect\Wpa\ when it has at most one reflex chain with at most one endpoint in~$\protect\P\setminus\protect\lhp{\protect\ea}$ and the distance from \protect\Tb\ to \protect\qa\ is greater than~$R$.}
			    \end{center}
			    \end{figure*}
Since \Wpa\ passes \qa\ and the remaining perimeter of \Da\ is convex (except where \Wpa\ follows the reflex chain), it follows that $\len{\Wpa}\leq\sup\len{\partial\Da}$ over the set of all possible tentacle cuts.

We bound the length of the perimeter of \Da\ as follows. Consider the smallest angle cone $\angle\Tb$ with apex \qa\ that contains \Tb. Let \ta\ and \tb\ be two points of \Tb\ that touch the sides of this cone; see Figure~\ref{circfig}(a). We assume a coordinate system where \qa\ is at the origin, \ta\ is at distance $R+y_1$, with $y_1>0$, directly above \qa\ and \tb\ lies in the first quadrant in this coordinate system at distance $R+y_2$, with $y_2>0$, from \qa. We can also assume that any reflex chain of \Wpa\ lies in the first quadrant in the coordinate system. It is clear that if the reflex chain has points in the fourth quadrant, the length of \Da\ can never attain its maximum possible ratio to $R$. Denote the part of \Da\ in the fourth quadrant (to the left of $\angle\Tb$) by \Dal. We can bound the length of the two circular arcs that form part of the perimeter of \Dal\ by
\begin{equation}
\len{\partial\Dal} = \pi R/2 - \alpha_1 R/2 + \alpha_1(R+y_1) = (\pi+\alpha_1)R/2 + \alpha_1\cdot y_1,
\end{equation}
where $\alpha_1=2\arcsin\big(R\big/(2R+2y_1)\big)$.

If we denote the part of \Da\ inside the cone $\angle\Tb$ by \Dac, the light green region in the cone $\angle\Tb$ in Figure~\ref{circfig}(a), we can bound the length of the circular arc perimeter of \Dac\ by $\len{\partial\Dac}\leq R\cdot\beta$, where $\len{[\ta,\tb]}^2 = (R+y_1)^2 + (R+y_2)^2 -2(R+y_1)(R+y_2)\cos\beta$ and $\beta$ is the angle of the cone $\angle\Tb$, by the cosine theorem; see Figure~\ref{circfig}(a). Again, we can assume that any reflex chain of \Wpa\ does not intersect the cone $\angle\Tb$ and therefore also not inside \Dac, otherwise the length of the perimeter of \Da\ can never attain its maximum possible ratio to~$R$.

If the remaining part of \Da\ intersects the single reflex chain of \Wpa, the length of the sequence of alternating circle arcs and line segments is bounded be the length of two line segments along two radii from \qa\ of length $R$ and \tb\ of length $R+y_2$, from their intersection point to their respective ends at \pa\ and \pb; see Figure~\ref{circfig}(a) and~(b). We denote these segments by $s_1$ and $s_2$ as in Figure~\ref{circfig}(b). Denote the angle between $s_1$ and $s_2$ by $\gamma$. Let \pc\ be the intersection point of the two circular arcs with center at \qa\ and length $R$ and center at \tb\ and length $R+y_2$ having positive $x$-coordinate in our coordinate system and let $l_1$ and $l_2$ be the subarcs of these two circular arcs from \pc\ to \pa\ and \pb, respectively; see Figure~\ref{circfig}(b). Let the radius $[\qa,\pa]$ have angle $\theta_1$ to the radius $[\qa,\pc]$ and let the radius $[\tb,\pb]$ have angle $\theta_2$ to the radius $[\tb,\pc]$. Thus, $\gamma=\pi/2-\alpha_2/2+\theta_1+\theta_2$.
Let $s'_1$ be the subsegment of $s_1$ outside the triangle $\bigtriangleup\qa,\tb,\pc$ and let $s'_2$ be the subsegment of $s_2$ outside $\bigtriangleup\qa,\tb,\pc$. The angle of the triangle at \pc\ is $\pi/2-\alpha_2/2\leq\gamma$ so $\len{s_1}-\len{s'_1}+\len{s_2}-\len{s'_2}\leq\len{l_1}+\len{s'_1}+\len{l_2}+\len{s'_2}$ giving $\len{s_1}+\len{s_2}\leq\len{l_1}+2\len{s'_1}+\len{l_2}+2\len{s'_2}$.
Let the two circle arcs symmetric to \Dal\ in the first quadrant, to the right of $\angle\Tb$ bound the region \Dar\ and we have
\begin{equation}
\len{\partial\Dar} =  (\pi+\alpha_2)R/2 + \alpha_2\cdot y_2,
\end{equation}
where $\alpha_2=2\arcsin\big(R\big/(2R+2y_2)\big)$; se Figures~\ref{circfig}(a) and~(c).
Define \Dax\ to be the union of \Dal, \Dac, and \Dar, i.e., the locus \Da\ when\Wpa\ does not have a reflex chain in the interior of \Da; see Figure~\ref{circfig}(c). Since $\len{s_1}+\len{s_2}\leq\len{l_1}+2\len{s'_1}+\len{l_2}+2\len{s'_2}$, we have that $\len{\partial\Da}\leq\len{\partial\Dax}+2\len{s'_1}+2\len{s'_2}$ and we bound the length of these three parts separately.

To maximize the length of the perimeter of \Dax\ we observe that the angle $\beta$ is as large as possible when the segment $[\ta,\tb]$ is as large as possible and the values $y_1$ and $y_2$ are as small as possible. The length of $[\ta,\tb]$ is bounded above by $R$ and $y_1$ and $y_2$ are minimal if the segment $[\ta,\tb]$ is almost tangent to the circle of radius $R$ centered at \qa. With this setup we can write the length of the perimeter of \Dax\ as a function of $y_1$, where $0\leq y_1\leq(\sqrt{2}-1)R$, and using standard variational calculus we obtain that
\begin{equation}\label{eqn:maxDax}
\sup_{0\leq y_1\leq(\sqrt{2}-1)R}\!\!\!\!\!\!\!\!\len{\partial\Dax(y_1)} < \left(\pi+2\sqrt{5}\arcsin\frac{1}{\sqrt{5}}\right)R,
\end{equation}
which occurs when $y_1=y_2=(\sqrt{5}/2-1)R$.

To maximize the lengths of $s'_1$ and $s'_2$, we realize that they are the longest when the radius $[\qa,\pa]$ intersects the segment $[\tb,\pc]$ perpendicularly and when the radius $[\tb,\pb]$ intersects the segment $[\qa,\pc]$ perpendicularly, i.e., when $\theta_1=\theta_2=\alpha_2/2$. Hence, the two segments are the longest when $\alpha_2$ is as large as possible, which occurs when the triangle $\bigtriangleup\qa,\tb,\pc$ is equilateral, whereby
\begin{equation}\label{eqn:maxseg}
\sup_{0\leq\alpha_2\leq\pi/6}\!\!\!\!\len{s'_1} =\!\!\!\!\sup_{0\leq\alpha_2\leq\pi/6}\!\!\!\!\len{s'_2} = (1-\sqrt{3}/2)R.
\end{equation}
This gives us a final bound of
\begin{align}\label{eqn:maxDa}
\len{\Wpa}
&<
\left(\pi+2\sqrt{5}\arcsin1/\sqrt{5}+4-2\sqrt{3}\right)R 
\nonumber\\
&=
\left(\pi/2+\sqrt{5}\arcsin1/\sqrt{5}+2-\sqrt{3}\right) \maxlen{(\Ta,\Tb)}.
\end{align}
\end{description}

We have from the two cases above, Inequalities~(\ref{eqn:opttour}) and~(\ref{eqn:maxDa}), that 
\begin{align}\label{eqn:maxWpa}
\!\!\!\!\len{\Wpa} 
&\leq
\max
\left\{
\begin{array}{l}
2\sqrt{2}R+\sqrt{2}\len{\Tb} = 2\sqrt{2}\maxlen{(\Ta,\Tb)} \approx 2.8284\maxlen{(\Ta,\Tb)},
\\
\!\!\left(\pi/2+\sqrt{5}\arcsin1/\sqrt{5}+2-\sqrt{3}\right)\maxlen{(\Ta,\Tb)} \approx 2.8754\maxlen{(\Ta,\Tb)}.
\end{array}
\right.
\end{align}
Hence, we choose the larger of the two values as our bound for the length of~\Wpa.

\end{description}

From Inequalities~(\ref{eqn:semicircle}), (\ref{eqn:nonintersection}), and (\ref{eqn:maxWpa}), we have for this case
\begin{align}\label{eqn:caseC}
\len{\Wa} 
&\leq
\len{\Wpa} + 2R + \len{\Ca}
\leq
\left(\pi+2\sqrt{5}\arcsin1/\sqrt{5}+4-2\sqrt{3}\right)R + 2R + 4\pi R/3
\nonumber\\&
=
\big(7\pi/3 + 6 - 2\sqrt{3} + 2\sqrt{5}\arcsin1/\sqrt{5}\big) R
\nonumber\\&
=
\big(7\pi/6 + 3 - \sqrt{3} + \sqrt{5}\arcsin1/\sqrt{5}\big) \maxlen{(\Ta,\Tb)}.
\end{align}

\end{description}

From the three cases that we have dealt with, Inequalities~(\ref{eqn:caseA}),~(\ref{eqn:caseB}), and~(\ref{eqn:caseC}), we have
\begin{align}
\len{\Wa}
&\leq
\max
\left\{
\begin{array}{l}
(\pi+1)\maxlen{(\Ta,\Tb)},
\\
4\maxlen{(\Ta,\Tb)},
\\
(7\pi/6 + 3 - \sqrt{3} + \sqrt{5}\arcsin1/\sqrt{5})\maxlen{(\Ta,\Tb)}.
\end{array}
\right.
\end{align}

The last case has the worst upper bound $\len{\Wa}\leq(7\pi/6 + 3 - \sqrt{3} + \sqrt{5}\arcsin1/\sqrt{5})\maxlen{(\Ta,\Tb)} \approx 5.969\maxlen{(\Ta,\Tb)}.$

Making the same argument for~\Wb, we obtain the same bound from which the lemma then follows.
\end{proof}

We note that we have made a slight overestimation when bounding the length of the perimeter of \Da\ as we are maximizing the length of the perimeter of \Dax\ and the lengths of $s'_1$ and $s'_2$ separately, Inequalities~(\ref{eqn:maxDax}) and~(\ref{eqn:maxseg}).
Some numerical experimentation indicates that our overestimation only affects the overall approximation bound in the second decimal, leading us to not pursue any improvement further.

\subsection{Complexity Analysis of the Algorithm}\label{sec:complexity}

The complexity analysis of the algorithm is straightforward. The for-loop in Step~\ref{alg:twr:approx} considers $O(n^2)$ pairs of extensions. Computing the tentacles and bases in Step~\ref{alg:twr:approx:bases} takes $O(n^6)$ time as established in Section~\ref{sec:computation}. The work in Step~\ref{alg:twr:approx:reduce} takes $O(n\log n)$ time since it is dominated by sorting the tentacles by length. Step~\ref{alg:twr:approx:compute} requires linear time using the algorithm by Toussaint~\cite{Tou:relconvexhull} and the test in Step~\ref{alg:twr:approx:test} takes constant time. Hence, the total time complexity for the algorithm is~$O(n^{8})$.

\begin{theorem}\label{thm:approx}
The {\it Two-Watchman-Route\/} algorithm computes a $5.969$-approx\-ima\-tion of the minmax two-watchman route and a $11.939$-approx\-ima\-tion of the minsum two-watchman route
in $O(n^8)$ time.
\end{theorem}

\section{A Trade-off between Computation Time and Accuracy}\label{sec:tradeoff}

We can trade the approximation bound for computational efficiency by realizing that relaxing the computation as described in Section~\ref{sec:computation} still provides base positions \qa\ and \qb\ sufficiently close to their optimal positions on \ea\ and \eb.
To do this, we reduce the computation in Section~\ref{sec:computation} to only use the first two cases, taking $O(n^2)$ time rather than $O(n^6)$ time. Let \qa\ be the base placement on \ea\ using the reduced computation and let \qsa\ be its optimal placement on \ea. Let \ZR(\qa,\rr,\be) be a longest tentacle in the jellyfish \MJFs(\qa). 
Now, consider the tentacle \ZR(\qsa,\rs,\be), taking into account that \rr\ moves to \rs\ on \be\ when the head moves from \qa\ to \qsa\ on \ea. 
There must exist a boundary edge \be'\ such that the tentacle \ZR(\qsa,\rr',\be') intersects the line through \qa\ orthogonal to \ea, otherwise \qa\ would be closer to \qsa; see Figure~\ref{tradeoff}. Hence, $\len{\qa^{},\qsa}\leq\len{\ZR(\qsa,\rr',\be')}$ and we have 
\begin{align}\label{eqn:tradeoff}
\len{\ZR(\qa,\rr,\be)}\leq\len{\ZR(\qa,\rs,\be)}\leq\len{\ZR(\qsa,\rs,\be)}+\len{\qa^{},\qsa}\leq\len{\ZR(\qsa,\rs,\be)}+\len{\ZR(\qsa,\rr',\be')}\leq2R,
\end{align}
since $R$ is the length of the longest tentacle in~\MJFs(\ea)\!.
We can argue similarly for \MJFs(\qb) giving us the jellyfish pair~\MJF(\qa,\qb) for which we can compute the reduced jellyfish pair and then its relative convex hulls.
\begin{figure*} 
			    \begin{center} \small
			    \input{\figdirtradeoff.pdf_t}

			    \parcaption{\capw}{\label{tradeoff}Illustrating the trade-off between the optimal base placement and the approximate base placement.}
			    \end{center}
			    \end{figure*}

Using $R'=2R$ in the proof of Lemma~\ref{lem:approx} instead of $R$, all the arguments go through giving us a tour with an approximation ratio at worst twice that of Theorem~\ref{thm:approx}. We state this as a corollary.

\begin{cortheorem}\label{thm:approx2}
The simplified {\it Two-Watchman-Route\/} algorithm computes a $11.939$-approx\-ima\-tion of the minmax two-watchman route and a $23.879$-approx\-ima\-tion of the minsum two-watchman route
in $O(n^4)$ time.
\end{cortheorem}

\section{Computing the Fixed Two-Watchman Route}\label{sec:fixed}

Since the heads \qa\ and \qb\ are given as input in this simpler case of the problem, it suffices to compute the Jellyfish pairs \JF(\qa,\qb) with \qa\ and \qb\ as heads which takes $O(n^2)$ time as explained in Section~\ref{sec:computation} and since \Wa\ intersects \Ta\ and \Wb\ intersects \Tb, the last case in the proof of Lemma~\ref{lem:approx} cannot occur, whereby the approximation bound for \Wa\ becomes
\begin{equation}
\len{\Wa} 
\leq
\len{\Wpa} + 2R + \len{\Ca}
\leq
2\sqrt{2}R+\sqrt{2}\len{\Tb}
+
2R
+
4\pi R/3
= (2\sqrt{2}+1+2\pi/3)\maxlen{(\Ta,\Tb)}.
\end{equation}
However, we note that this tour does not necessarily pass through \qa, unless \qa\ lies on the polygon boundary, hence we have to account for an extra reflex chain of length at most $2R$ to guarantee this. Making the same analysis for \Wb, the approximation factor becomes $2\sqrt{2}+2+2\pi/3\approx6.922$ in this case. The time complexity is dominated by computing the jellyfish, giving us the following theorem. 

\begin{theorem}\label{thm:fixed}
The algorithm computes a $6.922$-approx\-ima\-tion of the fixed minmax two-watchman route and a $13.845$-approx\-ima\-tion of the fixed minsum two-watchman route
in $O(n^2)$ time given two starting points for the two tours. If both starting points lie on the boundary, the approximation factors are $5.922$ and $11.845$, respectively.
\end{theorem}

\section{Conclusions}\label{sec:conc}
We have shown a polynomial time algorithms for computing constant factor approximations for the minmax and minsum two-watchman route in a simple polygon.


Our algorithms rely heavily on the fact that for two tours it is sufficient to guarantee that the boundary is seen to ensure that the complete polygon is seen. This does not hold for three or more tours. Thus, our method for the two-watchman tours does not generalize to the problem for three or more watchmen. Solving these problems within a constant factor bound remains elusive.

Establishing the complexity for the minsum two-watchman route is still open although our algorithm provides a polynomial time $11.939$-approximation.

The authors would like to thank Doc.~Åse Jevinger and Prof.~Pawe\l~\.Zyli\'nski for fruitful discussions that have improved the text immensely.



{
\small
\bibliographystyle{sccaptitleplain}
\bibliography{artgallery,cg,covers,geometry,holes,inapprox,linkcenter,rayshoot,rectgallery,shortpath,triangulation,%
visibility,watchmen,selfreference}
}

\appendix
\section{Proof of Lemma~\ref{lem:tentaclemotion}}\label{app:tentaclemotion}

{\bf Lemma.}~{\em
Let \q\ move a distance $\delta$ to \q'\ on a line segment $s$ and let \rr\ move a distance $\epsilon$ to \rr', where both \rr\ and \rr'\ lie in the open interval $]\v,\v'[$ of a boundary edge $\be=[\v,\v']$, in such a way that the first segment of the tentacles from \q\ and \q'\ intersect the same reflex vertex, if the tentacle consists of multiple segments, and \cut{\ZT(\q,\rr)} and \cut{\ZT(\q',\rr')} have the same hiding vertex, then
$$
\len{\ZT(\q',\rr')}  
=
\len{\ZT(\q,\rr)} +
{\cal F}(\delta,\epsilon),
$$
such that 
\begin{align*}
{\cal F}(\delta,\epsilon) & = 
- F_0 + \sqrt{F_0^2 + F_1\delta + F_2\delta^2}
- F_3 + \frac{F_3+F_4\epsilon+F_5\delta+F_6\epsilon\delta+F_7\epsilon^2+F_8\epsilon^2\delta}{\sqrt{1+F_9\epsilon+F_{10}\epsilon^2+F_{11}\epsilon^3+F_{12}\epsilon^4}}
\nonumber\\
& \qquad
- F_{13} + \sqrt{\frac{F_{13}^2+F_{14}\epsilon+F_{15}\delta+F_{16}\epsilon\delta+F_{17}\epsilon^2+F_{18}\delta^2+F_{19}\epsilon^2\delta+F_{20}\epsilon\delta^2+F_{21}\epsilon^2\delta^2}{1+F_{22}\epsilon+F_{23}\epsilon^2}},
\end{align*}
where $F_0,\dots,F_{23}$ are constants.
}

We begin with the assumption that \ZT(\q,\rr) consists of at least two segments whose end points touch reflex boundary vertices of the polygon. Since this is the case, we can separate the motion from \rr\ to \rr'\ on \be\ and the motion from \q\ to \q'\ on $s$ and handle those two cases separately. We first look at the motion from \rr\ to \rr'\ on \be\ and since a small change of \q\ on $s$ does not affect the motion on \be. 

Given a tentacle \ZT(\q,\rr), we emulate a point sliding from \rr\ to \rr'\ along \be. If the tentacle tip intersects the boundary of the polygon at a point \p, where the tentacle sees \rr, then we let the segment $[\rr,\rr'']$ be parallel to $[\p,\p']$ and intersecting the line collinear to $[\uh,\rr']$ at \rr'', where \uh\ is the hiding vertex; see Figure~\ref{slideproofr}(a). The lines through $[\p,\p']$ and $[\v,\v']$ intersect at a point $t$ (unless they are parallel, in which case we assume $t$ to be a point at infinity). Since the triangles $\bigtriangleup\rr'\!,\rr,\rr''$ and $\bigtriangleup\rr'\!,t,\p'$ are similar, we have
\begin{align*}
\frac{\epsilon}{\len{\rr,\rr''}} &= \frac{\len{\rr,t}-\epsilon}{\len{\p,t}+\epsilon'}
\end{align*}
and since the triangles $\bigtriangleup\rr,\uh,\rr''$ and $\bigtriangleup\p,\uh,\p'$ are also similar, we have
\begin{align*}
\frac{\len{\rr,\rr''}}{\len{\uh,\rr}} &= \frac{\epsilon'}{\len{\uh,\p}},
\end{align*}
giving us that 
\begin{align}\epsilon'=\dfrac{\len{\p,t}\cdot\len{\uh,\p}\cdot\epsilon}{\len{\rr,t}\cdot\len{\uh,\rr} - \left(\len{\uh,\p}+\len{\uh,\rr}\right)\cdot\epsilon} = \frac{a\cdot\epsilon}{1-c\cdot\epsilon},\label{eqn:eps to eps prime}
\end{align}
for constants $a$ and $c$ depending on the points \p, \rr, \uh, and~$t$.

Thus, the length of the tentacle as the point moves locally from \rr\ to \rr'\ along \be\ is 
\begin{align}
\len{\ZT(\q,\rr')} & = 
\len{\ZT(\q,\rr)} - \len{\u,\p} + \len{\u,\p'}
\nonumber\\
& =
\len{\ZT(\q,\rr)} - \len{\u,\p} + \sqrt{\len{\u,\p}^2 + \epsilon'^2 - 2\len{\u,\p}\epsilon'\cos\phi}
\nonumber\\
& = 
\len{\ZT(\q,\rr)} - \len{\u,\p} +\frac{\sqrt{
\len{\u,\p}^2(1-c\epsilon)^2 + a^2 \epsilon^2 - 2\len{\u,\p}a\epsilon\cos\phi(1-c\epsilon)
}
}{1 - c\epsilon}
\nonumber\\
& =
\len{\ZT(\q,\rr)} - A_0 + \frac{\sqrt{A_0^2+A_1\epsilon+A_2\epsilon^2}}{1+A_3\epsilon},
\label{eqn:r slidetouch}
\end{align}
for constants $A_0,\ldots,A_3$ that only depend on the points \u, \p, \rr, \uh, $t$, and the angle $\phi$; see Figure~\ref{slideproofr}(a). 
\begin{figure*} 
			    \begin{center} \small
			    \input{\figdirslideproofr.pdf_t}

			    \parcaption{\capw}{\label{slideproofr}Showing the motions of a tentacle along the interior of~\protect\be.}
			    \end{center}
			    \end{figure*}

If the tentacle tip does not touch the boundary of the polygon; see Figure~\ref{slideproofr}(b), then we realize that as the point moves locally from \rr\ to \rr', the length of the tentacle changes as
\begin{align}
\len{\ZT(\q,\rr')} & = 
\len{\ZT(\q,\rr)} - \len{\u,\p} + \len{\u,\p'}
\nonumber\\
&=
\len{\ZT(\q,\rr)} - \len{\u,\p} + \len{\u,\uh}\sin(\gamma-\alpha)
\nonumber\\
& =
\len{\ZT(\q,\rr)} - \len{\u,\p} + \len{\u,\uh}(\sin\gamma\cos\alpha - \cos\gamma\sin\alpha)
\nonumber\\
& =
\len{\ZT(\q,\rr)} - \len{\u,\p} + \len{\u,\uh}\sin\gamma\frac{\len{\uh,\rr}-\epsilon\cos\beta}{
\sqrt{\len{\uh,\rr}^2 -2\epsilon\len{\uh,\rr}\cos\beta +\epsilon^2}}
\nonumber\\
& \qquad\qquad
- \len{\u,\uh}\cos\gamma\frac{\epsilon\sin\beta}{\sqrt{\len{\uh,\rr}^2 -2\epsilon\len{\,\rr}\cos\beta +\epsilon^2}}
\nonumber\\
& =
\len{\ZT(\q,\rr)} - \len{\u,\p} + \frac{
\len{\u,\uh}\len{\uh,\rr}\sin\gamma - \epsilon\len{\u,\uh}(\cos\beta\sin\gamma + \cos\gamma\sin\beta)
}{\sqrt{\len{\uh,\rr}^2 -2\epsilon\len{\uh,\rr}\cos\beta +\epsilon^2}}
\nonumber\\
& =
\len{\ZT(\q,\rr)} - \len{\u,\p} + \frac{
\len{\u,\p}\len{\uh,\rr} - \epsilon\len{\u,\uh}\sin(\beta+\gamma)}{\sqrt{\len{\uh,\rr}^2 -2\epsilon\len{\uh,\rr}\cos\beta +\epsilon^2}}
\nonumber\\
& =
\len{\ZT(\q,\rr)} - B_0 + \frac{B_0+B_1\epsilon}{\sqrt{1+B_2\epsilon+B_3\epsilon^2}}
\label{eqn:r slidefree}
\end{align}
for constants $B_0,\ldots,B_3$ that only depend on the points \u, \p, \rr, \uh\ and the angles $\beta$ and $\gamma$; see Figure~\ref{slideproofr}(b).

Let us now consider the change that happens as the head moves a distance $\delta$ from \q\ to \q'\ on $s$, still under the assumption that \ZT(\q,\rr) and \ZT(\q',\rr') consists of at least two segments where the end points different from \q\ (\q') and \rr\ (\rr') touch reflex vertices of the polygon thus making turns at the boundary.
%
The length of the tentacle changes as
\begin{align}
\len{\ZT(\q',\rr)} & = 
\len{\ZT(\q,\rr)} - \len{\q,\u} + \len{\q',\u}
\nonumber\\
& =
\len{\ZT(\q,\rr)} - \len{\q,\u} + \sqrt{\len{\q,\u}^2 + \delta^2 -2\len{\q,\u}\delta\cos(\pi-\theta)}
\nonumber\\
& =
\len{\ZT(\q,\rr)} - C_0 + \sqrt{C_0^2+C_1\delta+C_2\delta^2}
,\label{eqn:head many}
\end{align}
for constants $C_0$, $C_1$ and $C_2$ that only depend on the points \u, \q, and the angle $\theta$; see Figure~\ref{slideproofq}(a).%
\begin{figure*} 
			    \begin{center} \small
			    \input{\figdirslideproofq.pdf_t}

			    \parcaption{\capw}{\label{slideproofq}Showing the motion of a tentacle along line segment~\protect\e.}
			    \end{center}
			    \end{figure*}

\medskip
If the \ZT(\q,\rr) and \ZT(\q',\rr') each consists of a single segment, then either the tip touches the boundary or it does not. 
If the tentacle \ZT(\q,\rr) is a single segment that intersects \cut{\ZT(\q,\rr)} without the tip touching the boundary, then 
\len{\q',\p'} depends on both $\epsilon$ and $\delta$. The length \len{\q',\p'} is established from \len{\q,\p} by first computing \len{\q,\ph}, a case we have already solved; see Equality~\ref{eqn:r slidefree} and Figure~\ref{slideproofr}(b); and then computing \len{\q',\p'} from \len{\q,\ph}. $\len{\q',\p'}=\len{\q,\ph}(1-\delta/\len{\q,t'})$ since $\len{\q',\p'}/(\len{\q,t'}-\delta)=\len{\q,\ph}/\len{\q,t'}$ by similarity.
Hence, the length of the tentacle changes as
\begin{align}
\len{\ZT(\q',\rr')} & = 
\len{\ZT(\q,\rr)} - \len{\q,\p} + \len{\q',\p'}
\nonumber\\
& =
\len{\ZT(\q,\rr)} - \len{\q,\p} + \left(1-\frac{\delta}{\len{\q,t'}}\right)\len{\q,\ph}
\nonumber\\
& =
\len{\ZT(\q,\rr)} - \len{\q,\p} + \left(1-\frac{\delta}{\len{\q,t}-\epsilon'}\right)\len{\q,\uh}\sin(\gamma-\alpha)
\nonumber\\
& =
\len{\ZT(\q,\rr)} - \len{\q,\p} + \left(1-\frac{\delta}{\len{\q,t}-\frac{a\epsilon}{1-c\epsilon}}\right)\len{\q,\uh}(\sin\gamma\cos\alpha - \sin\gamma\cos\alpha)
\nonumber\\
& =
\len{\ZT(\q,\rr)} - \len{\q,\p} + \mbox{}
\nonumber\\
& \qquad\qquad
+ \left(1-\frac{\delta}{\len{\q,t}-\frac{a\epsilon}{1-c\epsilon}}\right)\len{\q,\uh}\sin\gamma\frac{\len{\uh,\rr}-\epsilon\cos\beta}{\sqrt{\len{\uh,\rr}^2-2\epsilon\len{\uh,\rr}\cos\beta+\epsilon^2}}
\nonumber\\
& \qquad\qquad
- \left(1-\frac{\delta}{\len{\q,t}-\frac{a\epsilon}{1-c\epsilon}}\right)\len{\q,\uh}\cos\gamma\frac{\epsilon\sin\beta}{\sqrt{\len{\uh,\rr}^2-2\epsilon\len{\uh,\rr}\cos\beta+\epsilon^2}}
\nonumber\\
& =
\len{\ZT(\q,\rr)} - \len{\q,\p} + \mbox{}
\nonumber\\
& \qquad\qquad
+ \left(1-\frac{\delta(1-c\epsilon)}{\len{\q,t}-(c\len{\q,t}+a)\epsilon}\right)
\frac{\len{\q,\p}\len{\uh,\rr}-\epsilon\len{\u,\uh}\sin(\beta+\gamma)}{\sqrt{\len{\uh,\rr}^2-2\epsilon\len{\uh,\rr}\cos\beta+\epsilon^2}}
\nonumber\\
& =
\len{\ZT(\q,\rr)} - D_0 + \frac{D_0+D_1\epsilon+D_2\delta+D_3\epsilon\delta+D_4\epsilon^2+D_5\epsilon^2\delta}{\sqrt{1+D_6\epsilon+D_7\epsilon^2+D_8\epsilon^3+D_9\epsilon^4}}
,\label{eqn:head single slidefree}
\end{align}
for constants $D_0,\ldots,D_9$ that only depend on the points \p, \q, \rr, and \uh, together with the angles $\beta$ and $\gamma$; see Figure~\ref{slideproofq}(b).%

If the tentacle \ZT(\q,\rr) is a single segment that intersects \cut{\ZT(\q,\rr)} with the tip touching the boundary, then the length \len{\q',\p'} is established from \len{\q,\p} by the cosine theorem. Hence, the length of the tentacle changes as
\begin{align}
\len{\ZT(\q',\rr')} & = 
\len{\ZT(\q,\rr)} - \len{\q,\p} + \len{\q',\p'}
\nonumber\\
& =
\len{\ZT(\q,\rr)} - \len{\q,\p} 
+ \sqrt{(\len{\p,t}-\epsilon')^2 + (\len{\q,t}-\delta)^2 - 2(\len{\p,t}-\epsilon')(\len{\q,t}-\delta)\cos\phi}
\nonumber\\
& =
\len{\ZT(\q,\rr)} - \len{\q,\p} + \mbox{}
\nonumber\\
&  \qquad
+ \sqrt{\left(\len{\p,t}-\frac{a\epsilon}{1-c\epsilon}\right)^2 + (\len{\q,t}-\delta)^2 - 2\left(\len{\p,t}-\frac{a\epsilon}{1-c\epsilon}\right)(\len{\q,t}-\delta)\cos\phi}
\nonumber\\
& =
\len{\ZT(\q,\rr)} - E_0 
+ \sqrt{\frac{E_0^2+E_1\epsilon+E_2\delta+E_3\epsilon\delta+E_4\epsilon^2+E_5\delta^2+E_6\epsilon^2\delta+E_7\epsilon\delta^2+E_8\epsilon^2\delta^2}{1+E_{9}\epsilon+E_{10}\epsilon^2}}
,\label{eqn:head single slidetouch}
\end{align}
for constants $E_0,\ldots,E_{10}$ that only depend on the points \p, \q, \rr, \uh\ and the angle $\phi$; see Figure~\ref{slideproofq}(c).%

\bigskip
Combining Equalities~(\ref{eqn:r slidetouch})--(\ref{eqn:head single slidetouch}), we obtain the equality
\begin{align}
\len{\ZT(\q',\rr')} & = 
\len{\ZT(\q,\rr)} - F_0 + \sqrt{F_0^2 + F_1\delta + F_2\delta^2}
- F_3 + \frac{F_3+F_4\epsilon+F_5\delta+F_6\epsilon\delta+F_7\epsilon^2+F_8\epsilon^2\delta}{\sqrt{1+F_9\epsilon+F_{10}\epsilon^2+F_{11}\epsilon^3+F_{12}\epsilon^4}}
\nonumber\\
& \qquad
- F_{13} + \sqrt{\frac{F_{13}^2+F_{14}\epsilon+F_{15}\delta+F_{16}\epsilon\delta+F_{17}\epsilon^2+F_{18}\delta^2+F_{19}\epsilon^2\delta+F_{20}\epsilon\delta^2+F_{21}\epsilon^2\delta^2}{1+F_{22}\epsilon+F_{23}\epsilon^2}},
\end{align}
for constants $F_0,\ldots,F_{23}$, as claimed.
We note further that since $F_0$, $F_3$, and $F_{13}$ represent actual distances, these must be non-negative.\hfill$\Box$

\end{document}